\documentclass[acmlarge]{acmart}
\makeatletter
\newcommand{\confshort}{\acmConference@shortname}
\newcommand{\conffull}{\acmConference@name}
\newcommand{\confdate}{\acmConference@date}
\newcommand{\confloc}{\acmConference@venue}
\AtBeginDocument{
  \fancypagestyle{firstpagestyle}{
    \fancyhead{}%
    \fancyfoot[C]{}%
  }
  \fancyhf{}
  \fancyhead[LO]{\@headfootfont\shorttitle}%
  \fancyhead[RE]{\@headfootfont\@shortauthors}%
  \fancyhead[LE]{\@headfootfont\footnotesize \confshort, \confdate, \confloc}%
  \fancyhead[RO]{\@headfootfont\footnotesize \confshort, \confdate, \confloc}%
  \fancyfoot[C]{}%
}

\makeatother

\AtBeginDocument{%
  }

\usepackage[table]{xcolor}    
\usepackage{subcaption}

\usepackage{enumitem}
\usepackage{url}

\usepackage{tabularx}
\usepackage{array}

\usepackage{multirow}
\usepackage{tcolorbox}
\usepackage{float}
\usepackage{booktabs}
\usepackage{tikz}
\usepackage{pgfplots}
\pgfplotsset{compat=1.18}
\usetikzlibrary{positioning}
\usepgfplotslibrary{groupplots}
\usepackage{chngcntr} 
\counterwithin{equation}{section}

\newcommand{\meanstd}[2]{#1\textsubscript{\scriptsize$\pm$#2}}
\usetikzlibrary{patterns}

\usepackage{amsmath}
\usepackage{wrapfig}
\usepackage{tikz}
\usetikzlibrary{arrows.meta,positioning,fit,calc,backgrounds}

\copyrightyear{2026}
\acmYear{2026}
\setcopyright{usgovmixed}
\acmConference[FAccT '26]{The 2026 ACM Conference on Fairness, Accountability, and Transparency}{June 25--28, 2026}{Montreal, QC, Canada}
\acmBooktitle{The 2026 ACM Conference on Fairness, Accountability, and Transparency (FAccT '26), June 25--28, 2026, Montreal, QC, Canada}
\acmDOI{10.1145/3805689.3812242}
\acmISBN{979-8-4007-2596-8/2026/06}

\usepackage{graphicx}
\usepackage{subcaption} 
\usepackage{caption}

\begin{document}


\title[Selective Counterfactual Consistency for Vertical Federated Learning]{Toward Individual Fairness Without Centralized Data: \\Selective Counterfactual Consistency for Vertical Federated Learning}

\author{Dawood Wasif}
\orcid{0000-0002-0513-5890}

\email{dawoodwasif@vt.edu}
\affiliation{%
  \institution{Virginia Tech}
  \city{Alexandria}
  \state{Virginia}
  \country{USA}
}

\author{Chandan K. Reddy}
\orcid{0000-0003-2839-3662}

\email{reddy@cs.vt.edu}
\affiliation{%
  \institution{Virginia Tech}
  \city{Alexandria}
  \state{Virginia}
  \country{USA}
}

\author{Terrence J. Moore}
\orcid{0000-0003-3279-2965}

\email{terrence.j.moore.civ@army.mil}
\affiliation{%
  \institution{U.S. Army Research Laboratory}
  \city{Adelphi}
  \state{Maryland}
  \country{USA}
}

\author{Jin-Hee Cho}
\orcid{0000-0002-5908-4662}

\email{jicho@vt.edu}
\affiliation{%
  \institution{Virginia Tech}
  \city{Alexandria}
  \state{Virginia}
  \country{USA}
}

\renewcommand{\shortauthors}{Wasif et al.}

\begin{abstract}
When algorithmic decisions depend on data distributed across institutions, how can we ensure that an individual’s outcome does not change arbitrarily based on a protected attribute? We study this question in vertical federated learning (VFL), where features are split across parties, sensitive attributes may be private, and proxies for protected characteristics can be scattered across institutional boundaries under strict privacy constraints. Our focus is on individual-level counterfactual stability, i.e., per-instance prediction consistency under protected-attribute interventions as formalized in the causal fairness literature, rather than group parity guarantees such as demographic parity or equalized odds. We propose \textsc{SCC-VFL}, a server-centric framework for enforcing selective counterfactual consistency (SCC) at the individual level in VFL. SCC-VFL operationalizes a given policy specification by combining three components: (i) differentially private, graph-free discovery of feature roles into non-descendants, policy-permitted mediators, and impermissible proxies using only a formally private sketch of the sensitive attribute, with a formal per-release privacy that does not extend to the full training pipeline; (ii) masked counterfactual generation that edits only mediators while fixing non-descendants and suppressing proxy leakage; and (iii) server-side enforcement via an SCC consistency loss that penalizes impermissible prediction changes under protected-attribute interventions. Across three real-world datasets spanning credit, healthcare, and criminal justice, SCC-VFL maintains or improves predictive accuracy while sharply reducing decision flip rates by up to 98\% relative to strong baselines. It also lowers attribute-inference attack success and improves robustness, demonstrating favorable utility-fairness-privacy trade-offs in realistic VFL deployments.
\end{abstract}





\maketitle

\section{Introduction}
\label{sec:intro}

\noindent \textbf{Do sensitive attributes change the decision?}
Algorithmic decision systems increasingly shape high-stakes outcomes in credit, hiring, healthcare, and criminal justice, directly determining who receives resources and opportunities. This raises expectations that decisions must be accurate, defensible, and socially legitimate \cite{barocas2023fairml}. A core requirement is fairness: An individual’s outcome should not change simply because a protected attribute such as race, gender, or age changes, or because the model relies on an impermissible proxy. Yet many deployed systems satisfy population-level fairness constraints while still permitting decision changes for specific individuals under hypothetical changes to protected attributes \cite{kusner2017counterfactual}. Although the community has developed a substantial fairness toolkit \cite{barocas2023fairml}, much of it emphasizes group-level criteria such as demographic parity and equalized odds \cite{feldman2015disparate,hardt2016equality}, alongside individual notions of treating similar individuals similarly \cite{dwork2012awareness}. These criteria can be mutually incompatible in practice, forcing explicit trade-offs that are often opaque to affected individuals \cite{kleinberg2017risk, binns2020apparent}. As accountability demands intensify, there is growing need for individual-level protections that stakeholders can understand, audit, and contest \cite{barocas2023fairml}. Following the causal fairness literature \cite{kusner2017counterfactual,chiappa2019path}, we use \emph{individual fairness} throughout this paper to denote per-instance counterfactual stability under a declared policy specification, distinct from group-level parity constraints such as demographic parity or equalized odds.

\noindent \textbf{Practical motivation.}
Consider a credit consortium where Bank~A holds account histories, Employer~B holds income records, and Bureau~C holds third-party risk indicators. These parties jointly train a lending model via VFL without centralizing sensitive data. But how can they ensure that a 25-year-old applicant receives a decision stable under a counterfactual change in age, holding policy-relevant qualifications constant? This is a legal requirement, not merely a technical preference: age discrimination in lending is prohibited~\cite{weerts2023algorithmic}, yet no single party observes all features, and the sensitive attribute may be privately held. Similar challenges arise wherever high-stakes decisions rely on data distributed across organizational boundaries, including healthcare, employment, and criminal justice.

\noindent \textbf{Data are split, but decisions are shared.}
As predictive data become increasingly distributed across organizations, data islands are becoming common \cite{yang2019federated}. Federated learning enables collaborative training without centralizing raw data \cite{mcmahan2017fedavg,kairouz2021advances}. While most work assumes horizontal partitioning, many real settings are vertically partitioned, with parties holding disjoint feature subsets for overlapping individuals \cite{yang2019federated,yang2023survey}. In VFL, party-local representations are aggregated into a joint predictor, but fairness becomes harder because proxies may be distributed across parties, counterfactual interventions may be implausible under partial visibility, and sensitive attributes may still leak through shared representations or updates \cite{melis2019exploiting,luo2021featureinference,yang2023survey,jin2021cafe,fu2022label}. Existing approaches such as statistical constraints, representation invariance, and adversarial debiasing often fail to cleanly separate permissible mediators from impermissible proxies, leading to under-protection or utility loss \cite{zemel2013fairrepr,zhang2018adversarialbias}.

\noindent \textbf{Fairness is about causal what-if questions.}
Fairness in high-stakes decision-making hinges on causal influence rather than correlation \cite{pearl2009causality, loftus2018causal}. A natural standard is counterfactual consistency: an individual’s outcome should remain unchanged under a hypothetical intervention on the protected attribute, holding fixed what should remain the same for that person \cite{kusner2017counterfactual}. Crucially, this framing recognizes that not all causal pathways are alike. Some pathways are impermissible, others may be acceptable under policy, and observed features may act as proxies that should not transmit sensitive influence \cite{kilbertus2017avoiding,zhang2018causalexplanation}. Operationalizing this reasoning typically presumes a structural causal model or a vetted causal graph~\cite{barocas2020hidden, kasirzadeh2021use}. In practice, such graphs are rarely available, are often contested, and are costly to maintain in multi-party settings subject to distribution shift.

\noindent \textbf{Our proposal: SCC-VFL.}
We introduce SCC-VFL, or Selective Counterfactual Consistency for Vertical Federated Learning, a training-time framework that targets individual-level counterfactual stability while respecting privacy and partial visibility. SCC-VFL integrates three components: (i) a graph-free, differentially private procedure to identify candidate non-descendants, permissible mediators, and impermissible proxies; (ii) party-local masked counterfactual generation that edits only policy-permitted mediators while suppressing proxy leakage; and (iii) server-side enforcement via a selective counterfactual consistency loss that penalizes impermissible prediction changes while preserving legitimate pathways. Figure~\ref{fig:sccvfl_overview_small} summarizes the VFL fairness challenge and the design of SCC-VFL.

\noindent \textbf{Key contributions.}
This paper makes four contributions toward individual-level fairness and auditability in VFL:

\begin{enumerate}[leftmargin=*, nosep, label=(\arabic*)]

\item \textbf{Problem formulation: selective counterfactual consistency for VFL.}
We define SCC for VFL: predictions should be stable under counterfactual changes to a sensitive attribute, with non-descendants fixed and only policy-permitted mediators allowed to vary.

\item \textbf{Graph-free, policy-aware mask discovery under differential privacy.}
We propose a DP, server-assisted mask discovery that partitions each party’s features into non-descendants (N), mediators (M), and proxies (P), thereby separating policy specification from enforcement.

\item \textbf{Masked counterfactual generation with leakage control.}
We design party-local masked generators that edit only mediators, preserve non-descendant identity, and reduce proxy leakage while sharing only representations.

\item \textbf{Server-centric enforcement and empirical validation.}
We enforce SCC at aggregation and evaluate across multiple domains, reducing individual flip rates (up to 98\%) while maintaining accuracy and lowering empirical sensitive leakage under inference attacks.

\end{enumerate}

\begin{figure}
\centering
\begin{tikzpicture}[
  font=\footnotesize,
  >=Latex,
  box/.style={draw=black!70, rounded corners, align=left, inner sep=4pt, line width=0.7pt},
  party/.style={box, minimum width=3.35cm, minimum height=0.55cm},
  server/.style={box, minimum width=2.55cm, align=center},
  stage/.style={box, minimum width=3.15cm, minimum height=0.95cm},
  qbox/.style={box, draw=black!65, fill=black!3},
  p1box/.style={party, draw=blue!60!black, fill=blue!20},
  p2box/.style={party, draw=blue!60!black, fill=blue!10},
  p3box/.style={party, draw=blue!60!black, fill=blue!5},
  hardbox/.style={box, draw=black!60, fill=black!2},
  srvbox/.style={server, draw=black!65, fill=black!5},
  s1box/.style={stage, draw=green!55!black, fill=green!50},
  s2box/.style={stage, draw=green!55!black, fill=green!40},
  s3box/.style={stage, draw=green!55!black, fill=green!20},
  outbox/.style={stage, draw=green!55!black, fill=green!10},
  arrowP1/.style={-Latex, draw=blue!60!black, line width=0.9pt},
  arrowP2/.style={-Latex, draw=blue!60!black, dashed, line width=0.9pt},
  arrowP3/.style={-Latex, draw=blue!60!black, loosely dotted, line width=1.0pt},
  arrow12/.style={-Latex, draw=green!55!black, line width=0.9pt},
  arrow23/.style={-Latex, draw=green!55!black, dashed, line width=0.9pt},
  arrow3o/.style={-Latex, draw=green!55!black, dotted, line width=1.0pt},
  arrowLR/.style={-Latex, draw=black!55, line width=1.0pt},
  panel/.style={draw=black!60, fill=black!1, rounded corners, inner sep=6pt, line width=0.8pt}
]

\node[qbox, align=center, minimum width=6.8cm] (q)
{\textbf{\em If the sensitive attribute $s$ changes to $s'$,}\\
 \textbf{\em does the prediction $\hat{y}$ change?}};

\node[p1box, anchor=north west, align=left, text width=5.7cm] (p1)
  at ($(q.south west)+(0,-0.40cm)$)
  {\textbf{P1:} local features $\mathbf{x}_1$ (may hold $s$)};
\node[p2box, anchor=north west, align=left, text width=5.7cm] (p2)
  at ($(p1.south west)+(0,-0.18cm)$)
  {\textbf{P2:} local features $\mathbf{x}_2$ (may include proxies)};
\node[p3box, anchor=north west, align=left, text width=5.7cm] (p3)
  at ($(p2.south west)+(0,-0.18cm)$)
  {\textbf{P3:} local features $\mathbf{x}_3$ (may include proxies)};

\node[srvbox, minimum height=1.95cm, anchor=west]
  (srv) at ($(p2.east)+(2.10cm,0)$)
  {\textbf{Server}\\[2pt]prediction $\hat{y}$};

\draw[arrowP1] (p1.east) -- (p1.east -| srv.west)
  node[pos=0.52, above=-1pt, fill=white, inner sep=1pt, font=\scriptsize]{$\mathbf{h}_1$};
\draw[arrowP2] (p2.east) -- (p2.east -| srv.west)
  node[pos=0.52, above=-1pt, fill=white, inner sep=1pt, font=\scriptsize]{$\mathbf{h}_2$};
\draw[arrowP3] (p3.east) -- (p3.east -| srv.west)
  node[pos=0.52, above=-1pt, fill=white, inner sep=1pt, font=\scriptsize]{$\mathbf{h}_3$};

\node[hardbox, anchor=north west, text width=10.2cm]
  (hard) at ($(p3.south west)+(0,-0.30cm)$)
  {\textbf{\em Hard in VFL:} $s$ may be private; proxies are split;
   na\"{\i}ve counterfactuals (CFs) can be off-support; shared signals may leak.\\[2pt]
   \textbf{Notations:} $s$ sensitive attribute; $\mathbf{x}_i$ party features;
   $\mathbf{h}_i$ shared representation; $\hat{y}$ prediction.
   $L_{\text{task}}$ is the main prediction loss (e.g., cross-entropy);
   $L_{\text{SCC}}$ is the counterfactual consistency loss that penalizes
   prediction changes between an input and its masked counterfactual;
   $\lambda$ balances task utility and SCC regularization in
   $L_{\text{task}}+\lambda L_{\text{SCC}}$.};

\begin{scope}[on background layer]
\node[panel, fit=(q)(p1)(p2)(p3)(srv)(hard)] (L) {};
\end{scope}

\node[font=\bfseries, anchor=south west]
  at ($(L.north west)+(0.10cm,0.08cm)$) {Problem};

\node[s1box, right=14mm of L, yshift=2.00cm, align=center] (s1)
{\textbf{(1) DP masks}\\$N/M/P$ via $\tilde{s}$};
\node[s2box, below=4mm of s1, align=center] (s2)
{\textbf{(2) Masked CF}\\edit $M$, fix $N$, guard $P$};
\node[s3box, below=4mm of s2, align=center] (s3)
{\textbf{(3) Server SCC}\\$L_{\text{task}}+\lambda L_{\text{SCC}}$};
\node[outbox, below=4mm of s3, align=center] (out)
{\textbf{Outcome}\\more stable $\hat{y}$\\lower leakage};

\draw[arrow12] (s1.south) -- (s2.north);
\draw[arrow23] (s2.south) -- (s3.north);
\draw[arrow3o] (s3.south) -- (out.north);

\begin{scope}[on background layer]
\node[panel, fit=(s1)(s2)(s3)(out)] (R) {};
\end{scope}

\node[font=\bfseries, anchor=south west]
  at ($(R.north west)+(0.10cm,0.08cm)$) {\textsc{SCC-VFL}};

\draw[arrowLR] (L.east) -- ($(R.west |- L.east)$);

\end{tikzpicture}
\caption{Overview of \textsc{SCC-VFL}. \emph{Left:} In VFL, features are distributed across parties and sensitive attributes may be private, making naive counterfactual edits unstable or leaky. \emph{Right:} \textsc{SCC-VFL} discovers feature roles via differentially private masks, edits only policy-permitted mediators, and enforces prediction stability at server, improving stability while reducing leakage.}
\label{fig:sccvfl_overview_small}
\end{figure}
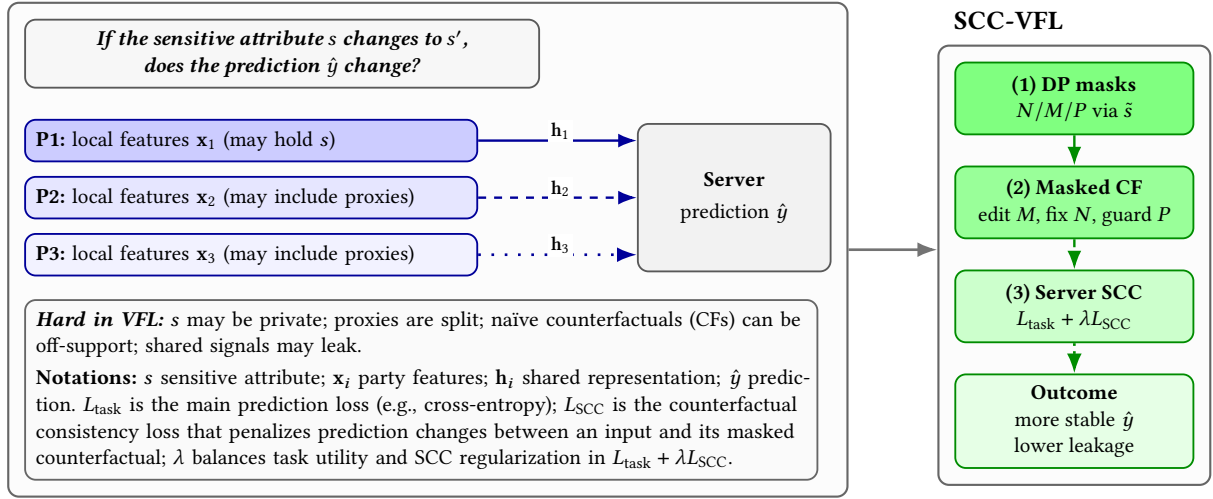

\section{Related Work}
\label{sec:related}

\subsection{Foundations of Algorithmic Fairness}
Foundational work formalizes algorithmic fairness through complementary group- and individual-level criteria, while exposing tensions among statistical parity goals, calibration, and utility \citep{hardt2016equality,barocas2016big,corbett2017algorithmic,kleinberg2017risk}. Group notions such as demographic parity and equalized odds target parity in outcomes or error rates across protected groups \citep{hardt2016equality,kleinberg2017risk}, while individual notions seek consistent treatment of similar individuals but typically depend on a task-specific similarity metric \citep{dwork2014algorithmic, yurochkin2019training}. Surveys and critiques emphasize that observational metrics can obscure causal pathways from sensitive attributes to predictions, and that satisfying one fairness criterion can preclude another \citep{chouldechova2017fair,zehlike2021fairness}. Although auditing and mitigation toolkits have matured, real deployments still report per-person inconsistencies even when aggregate metrics appear acceptable \citep{salazar2024survey}. These limitations motivate approaches that move beyond correlational parity toward individual-level guarantees and pathway-aware reasoning about how protected attributes and proxies influence decisions.

\subsection{Counterfactual and Causal Fairness}
Causal perspectives define fairness via interventions on protected attributes and by restricting which causal pathways may influence decisions. Building on counterfactual explanations as a contestability tool~\cite{wachter2017counterfactual}, counterfactual fairness deems a decision fair if it remains unchanged when the protected attribute is counterfactually altered, typically using structural causal models and a separation of descendant from non-descendant features~\cite{kusner2017counterfactual}. Path-specific fairness blocks impermissible routes while allowing policy-accepted mediators \citep{nabi2018fair,chiappa2019path}, and causal reinterpretations of equalized odds connect group criteria to interpretable mechanisms \citep{zhang2018equality}. Since full causal graphs are rarely available, prior work also studies partial-graph guarantees, safe feature sets, and graph-free surrogates  that reduce sensitive influence without hand-crafted causal structure \citep{zuo2022counterfactual,zhao2022towards}, along with practical approximations such as adversarial removal of sensitive signal and counterfactual augmentation \citep{madras2018learning,garg2019counterfactual}. Generative approaches further improve counterfactual validity by producing edits under identity and support constraints \citep{chiappa2019path,van2021decaf,sattigeri2019fairness}. Overall, these methods provide principled individual-level guarantees in centralized settings, but many assume full feature access or trusted causal structure, limiting applicability under privacy and partial observability constraints.

\subsection{Fairness in Federated and Vertical Federated Learning}
Federated learning (FL) raises fairness challenges beyond centralized training due to client heterogeneity, non-IID data, and privacy constraints \cite{liu2024vertical, wasif2025empirical}; existing work largely targets group fairness via reweighting, constrained or min--max objectives, and aggregation modifications to improve worst-group outcomes \citep{papadaki2022minimax,chu2021fedfair,ezzeldin2023fairfed,wasif2025resfl}, and studies joint group/individual formulations that surface fairness--utility trade-offs under non-IID regimes \citep{yue2023gifair}. Vertical federated learning (VFL) extends FL to vertically partitioned features for overlapping individuals, training via representation sharing and secure aggregation \citep{yang2023survey,bonawitz2017practical}, but fairness in VFL remains comparatively underexplored, with early constrained approaches imposing group-level objectives despite limited access to sensitive attributes \citep{liu2024vertical}. Bringing causal and counterfactual fairness into VFL is particularly challenging because mediators and impermissible proxies may be distributed across parties and the causal structure is rarely known or agreed upon \citep{kusner2017counterfactual,nabi2018fair,chiappa2019path, loftus2018causal}; although graph-free proxy screening and generative counterfactual modeling provide useful ingredients, most prior work is centralized or lacks per-instance guarantees under realistic VFL threat models \citep{van2021decaf,zuo2022counterfactual,sattigeri2019fairness}, motivating server-centric mechanisms that separate permissible from impermissible pathways, and promote stable decisions under protected-attribute interventions.

\section{Proposed Approach: SCC-VFL}
\subsection{VFL System Setup}

We consider VFL with $m$ parties over $n$ shared entities. Each party $p$ holds disjoint features $x_i^{(p)}\in\mathbb{R}^{d_p}$ and shares no raw data; the server aligns IDs, aggregates party representations, and hosts the prediction head. A protected attribute $s_i\in\mathcal{S}$, binary or multi-category, is held by one trusted party or processed in an enclave and is never revealed to the server, which accesses only privacy-preserving summaries when needed. The supervised label is $y_i$, and the full feature vector is $x_i=(x_i^{(1)},\ldots,x_i^{(m)})$. Table~\ref{tab:notation} summarizes the notation.

\begin{table}[t]
\centering
\caption{Consolidated notation. Symbols are grouped by scope: system-level, per-party, and training-related.}
\label{tab:notation}
\small
\setlength{\tabcolsep}{4pt}
\begin{tabular}{@{}lll@{}}
\toprule
\textbf{Symbol} & \textbf{Owner} & \textbf{Description} \\
\midrule
$n,\, m$ & System & Number of shared entities; number of parties \\
$x_i^{(p)} \!\in\! \mathbb{R}^{d_p}$ & Party $p$ & Local feature vector for entity $i$ \\
$s_i \!\in\! \mathcal{S}$ & Trusted holder & Protected (sensitive) attribute for entity $i$ \\
$y_i$ & Server & Supervised label (classification or regression) \\
$\phi^{(p)}$, $h_{\phi}^{(p)}$ & Party $p$ & Local encoder parameters and mapping \\
$\theta$, $f_{\theta}$ & Server & Prediction head parameters and mapping \\
$z_i \!\in\! \mathbb{R}^{d_z}$ & Server & Fused representation; $z_s\!=\!\psi(s)$: DP sketch of $s$ \\
$N^{(p)},\, M^{(p)},\, P^{(p)}$ & Party $p$ & Non-descendant, mediator, and proxy index sets \\
$x_i^{cf,(p)}$, $z_i^{cf}$ & Party $p$ / Server & Counterfactual input and fused representation \\
$\varphi^{(p)}$, $\omega^{(p)}$ & Party $p$ & Generator $g_{\varphi}^{(p)}$ and adversary $a_{\omega}^{(p)}$ parameters \\
$\lambda,\,\alpha',\,\beta',\,\eta'$ & Server & Loss weights (SCC, identity, support, leakage) \\
\bottomrule
\end{tabular}
\end{table}

Each party $p$ operates a local encoder $h_{\phi}^{(p)}$ parameterized by $\phi^{(p)}$; the server hosts a prediction head $f_{\theta}$ parameterized by $\theta$. The model follows split learning: each party sends an activation $h_{\phi}^{(p)}(x_i^{(p)})$ to the server, which fuses them into
\begin{equation}
z_i=\operatorname{Fuse}\!\big(h_{\phi}^{(1)}(x_i^{(1)}),\dots,h_{\phi}^{(m)}(x_i^{(m)})\big)\in\mathbb{R}^{d_z},
\label{eq:fuse}
\end{equation}
and predicts $f_{\theta}(z_i)\in\mathbb{R}^{k}$ ($k{=}1$ for regression; $k$ classes for classification).  To define counterfactual interventions on $s$, each party's coordinates are partitioned into non-descendants $N^{(p)}$ (fixed), permitted mediators $M^{(p)}$ (may change), and impermissible proxies $P^{(p)}$ (should not transmit sensitive influence). Under $s\!\leftarrow\!s'$, we construct masked counterfactual inputs $x_i^{cf,(p)}$ with $N^{(p)}$ unchanged, $M^{(p)}$ edited on-support for $s'$, and proxy leakage from $P^{(p)}$ suppressed; the full counterfactual is $x_i^{cf}=(x_i^{cf,(1)},\ldots,x_i^{cf,(m)})$ with fused representation $z_i^{cf}$. The server compares $f_{\theta}(z_i)$ and $f_{\theta}(z_i^{cf})$ and penalizes differences via a consistency term (Section~\ref{subsec:server}), while the supervised objective is
\begin{eqnarray}
\mathcal{L}_{\text{task}}(\theta,\phi)
=\frac{1}{n}\sum_{i=1}^{n}\ell\!\big(f_{\theta}(z_i),\,y_i\big),
\label{eq:ltask}
\qquad
z_i
=\operatorname{Fuse}\!\big(h_{\phi}^{(1)}(x_i^{(1)}),\dots,h_{\phi}^{(m)}(x_i^{(m)})\big).
\label{eq:zi_def}
\end{eqnarray}

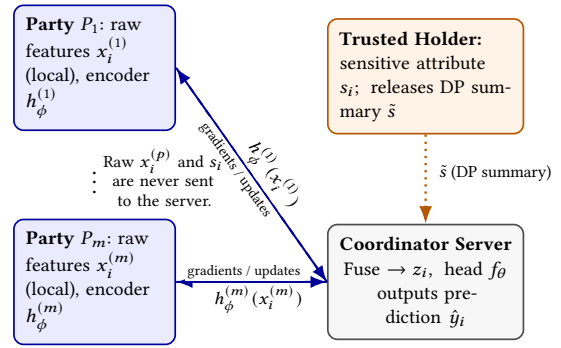
\begin{wrapfigure}{r}{0.45\textwidth}
\vspace{-2mm}
\centering
\resizebox{0.45\textwidth}{!}{%
\begin{tikzpicture}[
  font=\footnotesize,
  >=Latex,
  box/.style={draw=black!70, rounded corners, align=left, inner sep=5pt, line width=0.7pt},
  party/.style={box, text width=3cm, minimum height=1.05cm, draw=blue!60!black, fill=blue!8},
  server/.style={box, align=center, draw=black!70, fill=black!3},
  trusted/.style={box, draw=orange!70!black, fill=orange!10},
  fwd/.style={-Latex, draw=blue!60!black, line width=0.85pt},
  bwd/.style={-Latex, draw=blue!60!black, dashed, line width=0.85pt},
  aux/.style={-Latex, draw=orange!70!black, dotted, line width=0.95pt},
  lab/.style={font=\scriptsize, fill=white, inner sep=1pt}
]

\node[party, text width = 2cm] (p1) {\textbf{Party $P_1$}: raw features $x_i^{(1)}$ (local), encoder $h_{\phi}^{(1)}$};

\node[below=2.5mm of p1, font=\bfseries] (dots) {$\vdots$};

\node[party, below=2.5mm of dots, text width = 2cm] (pm) {\textbf{Party $P_m$}: raw features $x_i^{(m)}$ (local), encoder $h_{\phi}^{(m)}$};

\node[server, right=22mm of pm, text width=2.5cm] (srv)
{\textbf{Coordinator Server}\\Fuse $\rightarrow z_i$,\; head $f_{\theta}$\\outputs prediction $\hat{y}_i$};

\node[trusted, above=13mm of srv, text width=2.5cm] (T)
{\textbf{Trusted Holder:} sensitive attribute $s_i$;\; releases DP summary $\tilde{s}$};

\draw[fwd] (p1.east) -- node[pos=0.55, above, sloped, lab] {$h_{\phi}^{(1)}(x_i^{(1)})$} (srv.west);
\draw[fwd] (pm.east) -- node[pos=0.55, below, sloped, lab] {$h_{\phi}^{(m)}(x_i^{(m)})$} (srv.west);

\draw[bwd] (srv.west) -- node[pos=0.55, below, sloped, lab, font=\tiny] {gradients / updates} (p1.east);
\draw[bwd] (srv.west) -- node[pos=0.55, above, sloped, lab, font=\tiny] {gradients / updates} (pm.east);

\draw[aux] (T.south) -- node[pos=0.55, right, yshift=1.5mm, xshift=1.5mm, lab] {$\tilde{s}$ (DP summary)} (srv.north);

\node[font=\scriptsize, align=center, anchor=north, text width=2cm]
at ($(dots.south)+(1cm, 9mm)$)
{Raw $x_i^{(p)}$ and $s_i$ are never sent to the server.};

\end{tikzpicture}%
}
\caption{System model for VFL.}
\label{fig:vfl_system_model}
\vspace{-5mm}
\end{wrapfigure}

Figure~\ref{fig:vfl_system_model} illustrates the VFL system model, in which parties keep raw features local and exchange only intermediate representations with a coordinating server, while the sensitive attribute is held by a trusted party and accessed solely through a differentially private summary.

The remainder of this section introduces three components that together enable privacy-preserving, server-centric enforcement of individual-level counterfactual stability: selective mask discovery, masked counterfactual generation, and server-side consistency enforcement. We then describe the end-to-end privacy and security mechanisms that support deployment under realistic VFL threat models.

Specifically, SCC--VFL assumes a domain policy that divides the characteristics of each party into non-descendants \(N^{(p)}\) (held fixed under \(s\!\leftarrow\!s'\)), permissible mediators \(M^{(p)}\) (able to change on-support) and impermissible proxies \(P^{(p)}\) (descendants of \(s\) treated as disallowed pathways). Our contribution is enforcement given a declared policy, not deciding the policy: the DP-sketch screening in Section~\ref{subsec:mask} surfaces candidate descendants, while the final \(N/M/P\) assignment and its governance are specified and auditable via the protocol in Appendix~\ref{app:policy-audit}. Consistency is interpreted as promoting local decision stability under policy-permitted mediator edits (recourse-style changes), rather than claiming removal of all mediator pathway effects.

\subsection{Selective Descendant Discovery}
\label{subsec:mask}

Selective descendant discovery identifies, within each party, which local feature coordinates plausibly respond to interventions on the protected attribute after conditioning on other available information. For each party $p\in\{1,\dots,m\}$, it outputs a three-way mask $(N^{(p)},M^{(p)},P^{(p)})$ over local feature indices: non-descendants $N^{(p)}$, candidate mediators $M^{(p)}$, and stress-test proxies $P^{(p)}\subseteq M^{(p)}$. The goal is not causal identification, but a policy-aware, graph-free screening of features whose behavior is consistent with sensitivity to interventions on $s$. The procedure is graph-free and relies on conditional predictability evidence with a lightweight interventional validation. The server first produces a DP sketch $z_s=\psi(s)$, a low-dimensional encoding of $s$ with clipping and Gaussian noise calibrated to a privacy budget (Proposition~\ref{prop:dp_sketch}). Parties never receive raw $s$ and use only $z_s$ as an auxiliary covariate in local tests.

\begin{proposition}[DP guarantee for sketch release]
\label{prop:dp_sketch}
For a single feature coordinate $j$, let $c_j\in\mathbb{R}^{2\times K}$ be the contingency table of $s$ versus the discretised feature, and let neighboring datasets $D,D'$ differ by the addition or removal of one record. The mechanism that (i) clips $c_j$ to $\ell_2$-sensitivity $S$ via $\bar{c}_j = c_j \cdot \min(1, S/\|c_j\|_2)$ and (ii) releases $\tilde{c}_j = \bar{c}_j + \mathcal{N}(0,\,\sigma_{\mathrm{sketch}}^2 S^2 I)$ satisfies $(\varepsilon,\delta)$-differential privacy with
\[
\varepsilon \;=\; \frac{1}{\sigma_{\mathrm{sketch}}}\sqrt{2\ln(1.25/\delta)},
\]
for any $\delta\in(0,1)$, under the standard Gaussian mechanism analysis~\cite{dwork2014algorithmic}.
\end{proposition}

\noindent\textbf{Scope of the DP claim.} Differential privacy applies \emph{only} to the released contingency-table sketch used for mediator/proxy screening (Eqs.~\eqref{eq:clip_counts}--\eqref{eq:gauss_mech} in Appendix~\ref{app:dp-secagg}). It does \emph{not} cover representations, gradients, or model parameters. We do not claim end-to-end DP for the full SCC-VFL training pipeline. Formal composition accounting across mask refreshes is deferred to future work.

For coordinate $j$ at party $p$, let $x_j^{(p)}$ be the target and $\tilde{x}_{\setminus j}^{(p)}$ the remaining local coordinates, optionally compressed by a small shared encoder. Party $p$ fits two predictors of $x_j^{(p)}$ from $\tilde{x}_{\setminus j}^{(p)}$: $\hat{g}^{(p)}_{\mathrm{w/}~z_s}$, which also takes $z_s$ as input, and $\hat{g}^{(p)}_{\mathrm{w/o}~z_s}$, which does not. The predictors share the same architecture. Using a held-out risk $\mathcal{R}$ (e.g., squared error or negative log-likelihood), the party computes a risk difference and a Hilbert-Schmidt Independence Criterion (HSIC) statistic \cite{gretton2005measuring}:
\begin{equation}
\underbrace{\Delta_j^{(p)}}_{\substack{\text{Risk gain} \\ \text{from adding }z_s}}
=
\underbrace{\mathcal{R}\big(\hat{g}^{(p)}_{\mathrm{w/}~z_s};\,x_j^{(p)} \mid \tilde{x}_{\setminus j}^{(p)}, z_s\big)}_{\text{Held-out risk with }z_s}
-
\underbrace{\mathcal{R}\big(\hat{g}^{(p)}_{\mathrm{w/o}~z_s};\,x_j^{(p)} \mid \tilde{x}_{\setminus j}^{(p)}\big)}_{\text{Held-out risk without }z_s},
\; \;  
\underbrace{\widehat{\mathrm{HSIC}}_j^{(p)}}_{\substack{\text{Residual} \\ \text{dependence on }z_s}}
=
\underbrace{\mathrm{HSIC}\big(r_j^{(p)}, z_s \mid \tilde{x}_{\setminus j}^{(p)}\big)}_{\text{Conditional contrast score}},
\label{eq:mask_scores}
\end{equation}
where $r_j^{(p)}$ are residuals from one predictor, for example $\hat{g}^{(p)}_{\mathrm{w/o}~z_s}$. Intuitively, $\widehat{\mathrm{HSIC}}_j^{(p)}$ measures residual dependence on $z_s$ after conditioning, while $\Delta_j^{(p)}$ captures how much predictability improves when $z_s$ is available.

\begin{tcolorbox}[title=Intuition: How Mask Discovery Works]
\textbf{Goal:} Identify which features at each party are statistically responsive
to interventions on sensitive attribute $s$.

\textbf{Challenge:} Most parties cannot directly observe $s$ due to privacy
constraints.

\textbf{Approach:} The server constructs a differentially private sketch $z_s$,
a blurred representation of $s$ that preserves aggregate patterns while protecting
individual values. Each party then tests whether access to $z_s$ improves prediction
of its local features, conditional on other available information.

\textbf{Outcome:} Features whose predictability changes when $z_s$ is available
are flagged as candidates for policy review. This ranking supports, but does not
determine, the $N/M/P$ partition used for counterfactual enforcement.
\end{tcolorbox}

The server collects $\big(\Delta_j^{(p)},\widehat{\mathrm{HSIC}}_j^{(p)}\big)$ across parties and forms the tri-partition by assigning small $\widehat{\mathrm{HSIC}}_j^{(p)}$ to $N^{(p)}$, large $\widehat{\mathrm{HSIC}}_j^{(p)}$ to $M^{(p)}$, and within $M^{(p)}$ designating proxies $P^{(p)}$ as those with large $|\Delta_j^{(p)}|$, indicating strong sensitivity to $s$ but limited task relevance. These scores serve as lightweight surrogates for full graph-based causal discovery and are intended to surface candidates for policy review rather than definitive causal roles. Finally, a lightweight interventional validation closes the loop. The server toggles $z_s$ to a target $z_{s'}$ and instructs parties to apply the masked generator. Mediators in $M^{(p)}$ are edited toward $s'$, while non-descendants in $N^{(p)}$ are copied exactly. If identity on $N^{(p)}$ is violated or mediators fail to respond, assignments are revised, yielding masks that are both statistically supported and intervention-consistent.

\subsection{Masked Counterfactual Generation}
\label{subsec:gen}

Given per-party masks $(N^{(p)},M^{(p)},P^{(p)})$, masked counterfactual generation edits only mediator coordinates in a manner consistent with the conditional distribution under a target sensitive embedding. The objective is to preserve identity on non-descendants and suppress proxy leakage while producing plausible, on-support edits. Each party $p$ trains a compact conditional generator $g_{\varphi}^{(p)}$ that takes local non-descendants $x_{N}^{(p)}$, current mediators $x_{M}^{(p)}$, a server-provided context vector $c$ summarizing cross-party information, and the target embedding $z_{s'}$. The generator outputs edited mediators $x_{M}^{cf,(p)}$ intended to lie on the support of $p(x_{M}^{(p)}\mid x_{N}^{(p)},c,z_{s'})$, while enforcing $x_{N}^{cf,(p)}=x_{N}^{(p)}$. Proxy features $x_{P}^{(p)}$ pass through unchanged but are guarded by an adversary that attempts to recover $z_s$ from $(x_{P}^{cf,(p)},x_{N}^{cf,(p)})$. The generators and encoders minimize the recovery success, thereby reducing the residual-sensitive signal without discarding features. Parties share only intermediate activations for $(x^{(p)},x^{cf,(p)})$ and scalar validity metrics with the server, where the SCC loss further discourages proxy-dependent prediction variation.

The per-party generator loss combines identity, support, and leakage-control terms. Writing $q_\varphi$ for a conditional variational autoencoder and $a_\omega^{(p)}$ for the $s$-adversary, we define
\begin{eqnarray}
\mathcal{L}_{\text{gen}}^{(p)}(\varphi,\omega)
&=&
\underbrace{\alpha \,\|x_{N}^{cf,(p)} - x_{N}^{(p)}\|_2^2}_{\text{Identity on non-descendants }N^{(p)}} \;+\;
\underbrace{\beta\, \mathbb{E}\!\left[-\log q_\varphi\!\left(x_{M}^{cf,(p)} \mid x_{N}^{(p)}, c, z_{s'}\right)\right]}_{\text{On-support mediator likelihood under }z_{s'}} \nonumber\\
&&\;+\;
\underbrace{\gamma\, \mathrm{MMD}\!\left(x_{M}^{cf,(p)}, \,\mathcal{D}_{M\mid N,c,z_{s'}}\right)}_{\text{Kernel two-sample support matching for mediators}}
\;-\;
\underbrace{\eta\, \mathbb{E}\!\left[\log\!\big(1 - a_\omega^{(p)}(x_{P}^{cf,(p)}, x_{N}^{cf,(p)})\big)\right]}_{\text{Proxy leakage suppression via }s\text{-adversary}},
\label{eq:lgen}
\end{eqnarray}
with nonnegative weights $(\alpha,\beta,\gamma,\eta)$. The first term enforces identity on non-descendants, the second and third promote on-support mediator edits, and the last suppresses residual sensitive predictability through proxies. An optional cycle-consistency term regenerates mediators from $z_{s'}$ back to $z_s$ to further stabilize edits. The generator operates only on mediator coordinates, limiting communication and computational overhead while helping preserve task accuracy. 

Appendix~\ref{app:worked_example} provides a worked credit example illustrating how mediator edits and SCC enforcement interact under a counterfactual intervention.

\begin{figure*}[t]
  \centering
  \includegraphics[width=0.9\textwidth]{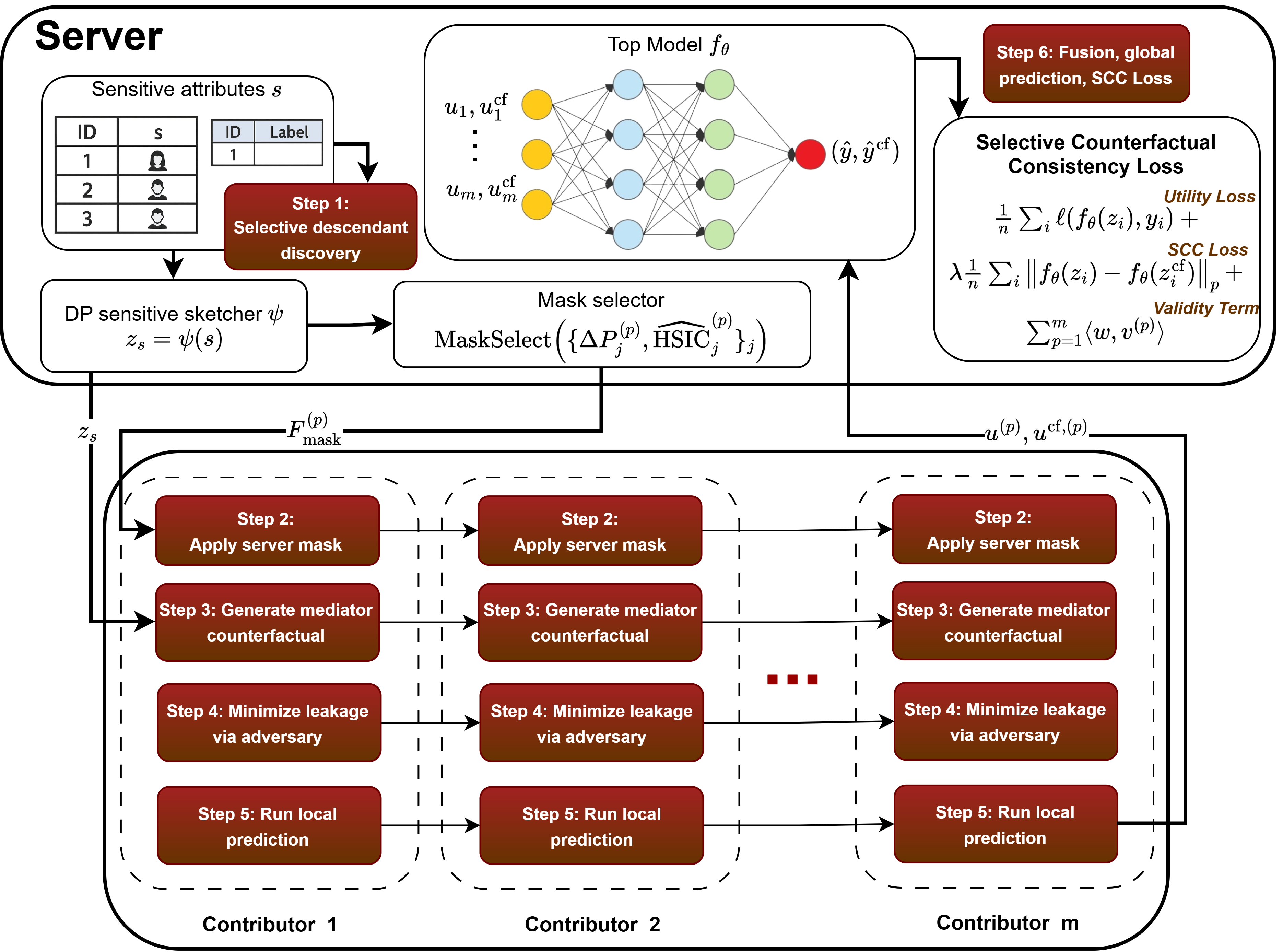}
  \caption{Overview of SCC--VFL. The server discovers descendants (Step 1), sketches sensitive attributes \(z_s\), selects a per-client mask via \(\Delta P_j^{(p)},\mathrm{HSIC}_j^{(p)}\), sends \(F_{\mathrm{mask}}^{(p)}\) to contributors for Steps 2--5, receives \(\{u^{(p)},u^{cf,(p)}\}\), and performs fusion with the SCC loss (Step 6).}
  \vspace{-1.5\baselineskip}
  \label{fig:sccvfl-server}
\end{figure*}

\subsection{Server-side Aggregation and Consistency}
\label{subsec:server}

Fairness is enforced at the server because it is the only point where the fused decision is formed. For each minibatch, the server instructs parties to compute local activations on both the original inputs and masked counterfactual inputs targeted to one or more sensitive alternatives. It then fuses the received activations to obtain $z$ and $z^{cf}$, evaluates the task loss on $z$, and applies a Selective Counterfactual Consistency (SCC) penalty that constrains prediction changes between $z$ and $z^{cf}$, measured on logits for classification and scalar outputs for regression. Importantly, SCC is applied to \emph{bounded, on-support} mediator edits produced by the generator, and is intended to prevent brittle decision flips under such permissible variability rather than to eliminate all mediator-driven effects.
Because counterfactuals differ only along policy-permitted mediators, with non-descendants fixed and proxies guarded, this penalty targets impermissible sensitive influence rather than suppressing legitimate causal pathways.

The server minimizes a joint objective combining supervised utility, per-example consistency, and per-party validity summaries: $\mathcal{V}^{(p)}_{\text{id}}$ (identity error on $N^{(p)}$), $\mathcal{V}^{(p)}_{\text{supp}}$ (mediator support adherence), and $\mathcal{V}^{(p)}_{\text{leak}}$ (proxy leakage measured by adversarial recovery). Formally,
\begin{equation}
\min_{\theta,\phi,\{\varphi^{(p)}\},\{\omega^{(p)}\}} \;
\underbrace{\mathcal{L}_{\text{task}}(\theta,\phi)}_{\substack{\text{Supervised utility}\\\text{on original inputs}}}
+
\underbrace{\lambda \frac{1}{n}\sum_{i=1}^n 
\big\|f_\theta(z_i)-f_\theta(z_i^{cf})\big\|_{p}}_{\substack{\text{SCC penalizing prediction}\\\text{changes under masked counterfactuals}}}
+
\underbrace{\sum_{p=1}^m \Big[
\alpha'\,\mathcal{V}^{(p)}_{\text{id}}
+ \beta'\,\mathcal{V}^{(p)}_{\text{supp}}
- \eta'\,\mathcal{V}^{(p)}_{\text{leak}}
\Big]}_{\substack{\text{Validity control: identity on $N^{(p)}$,}\\\text{on-support edits for $M^{(p)}$, reduced leakage via $P^{(p)}$}}},
\label{eq:server_obj}
\end{equation}

with weights $(\lambda,\alpha',\beta',\eta')$ chosen via validation on utility and stability metrics, where $p=1$ for classification (logits) and $p=2$ for regression.

The server backpropagates supervised and consistency gradients through the fusion operator and broadcasts split gradients to parties. Parties update local encoders and generators independently, while the server updates the prediction head, preserving separation of responsibilities and privacy constraints. The server also governs mask refresh. At fixed intervals, it recomputes descendant statistics under secure aggregation, applies false discovery rate control, and updates masks with hysteresis to prevent oscillations. 

\subsection{Privacy and Training Overview}
\label{subsec:privacy_impl}
SCC-VFL preserves sensitive information under the vertical federated threat model by never transmitting the protected attribute $s$ in the clear. The server computes a clipped and noisy sketch $z_s=\psi(s)$ that satisfies $(\varepsilon,\delta)$-DP for the released contingency tables (Proposition~\ref{prop:dp_sketch}; optionally computed inside a trusted execution environment when stronger isolation is required), and parties receive only $z_s$ and target variants $z_{s'}$, never raw sensitive values. 
\begin{wrapfigure}{r}{0.63\textwidth}
  \centering
  \vspace{-0.25\baselineskip}
  \includegraphics[width=\linewidth]{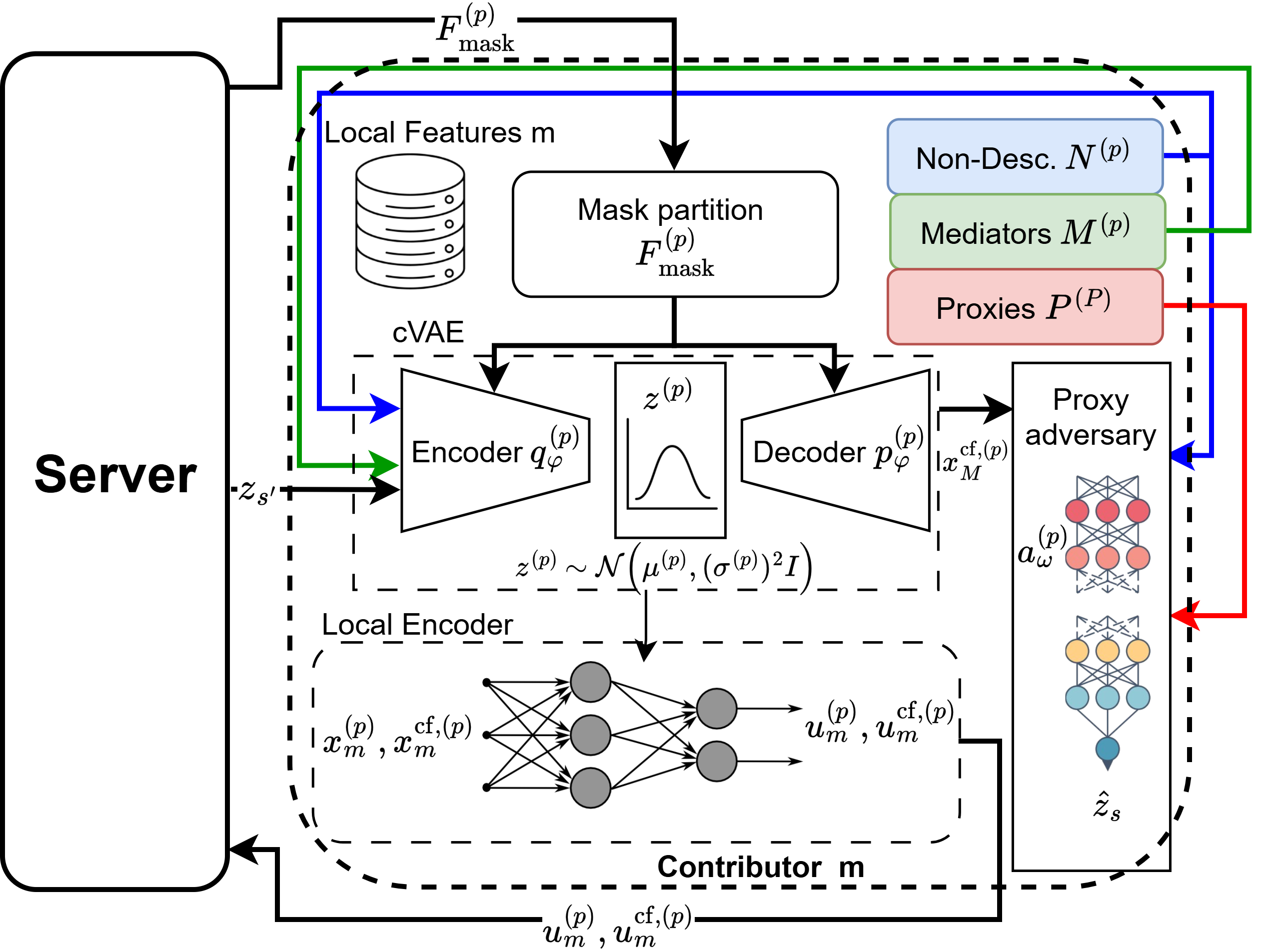}
  \caption{Contributor $m$ in SCC-VFL.}
  \label{fig:sccvfl-client}
  \vspace{-5mm}
\end{wrapfigure} 
We emphasize that this DP guarantee applies only to the sketch release and does not extend to representations, gradients, or model parameters exchanged during training. Descendant statistics and validity metrics are shared via secure aggregation, so the server observes only minibatch-level aggregates; only fixed-width representations (optionally randomized) cross authenticated channels. Audit logs record mask versions, privacy budgets, leakage metrics, and refresh events for governance review.
 
Training proceeds in three phases that may partially overlap: (1) selective descendant discovery for a few epochs with frozen encoders (securely aggregated statistics, FDR control, mask initialization), (2) party-local masked counterfactual generator training with identity, on-support, and leakage-control objectives while warm-starting the supervised head, and (3) joint optimization with the server-side Selective Counterfactual Consistency penalty, interleaving short fine-tuning rounds with periodic mask refresh events to stabilize the utility-stability trade-off. Hyperparameters are chosen from a narrow grid using validation curves of accuracy versus flip rate and consistency gap; for classification we compute consistency on pre-softmax logits (to avoid saturation), and for regression we use squared differences. Defining $c_i=\|f_\theta(z_i)-f_\theta(z_i^{cf})\|_p$ and $v^{(p)}=(\mathcal{V}^{(p)}_{\text{id}},\mathcal{V}^{(p)}_{\text{supp}},\mathcal{V}^{(p)}_{\text{leak}})$, the overall objective is
\begin{equation}
\min_{\Theta}\;\frac{1}{n}\sum_{i=1}^n \ell\!\big(f_{\theta}(z_i),y_i\big)
\;+\;\lambda\,\frac{1}{n}\sum_{i=1}^n c_i
\;+\;\sum_{p=1}^m\Big(\alpha'\mathcal{V}^{(p)}_{\text{id}}+\beta'\mathcal{V}^{(p)}_{\text{supp}}-\eta'\mathcal{V}^{(p)}_{\text{leak}}\Big),
\qquad
\Theta=\big(\theta,\phi,\{\varphi^{(p)}\},\{\omega^{(p)}\}\big).
\end{equation}

This decomposition keeps enforcement server-centric (stability on fused predictions) while parties ensure identity preservation on non-descendants, on-support mediator edits, and reduced proxy leakage. Figure~\ref{fig:sccvfl-server} summarizes the end-to-end SCC-VFL pipeline, and Figure~\ref{fig:sccvfl-client} illustrates Contributor $m$ (masking, mediator counterfactual generation, adversarial leakage suppression, and representation return for fusion).

\section{Experimental Setup}
\label{sec:evaluation}

The evaluation of SCC-VFL tests whether it improves individual-level counterfactual fairness while preserving utility and privacy in realistic VFL settings. We run experiments across three domains (banking, healthcare, and criminal justice), using consistent vertical partitions, standard metrics, and baselines to enable cross-domain comparison.

\subsection{Datasets}

\textbf{(1) Banking (German Credit, UCI):} We use German Credit \cite{hofmann1994statlog} with a binary default label. The sensitive attribute is age, binarized at 25 and used only for fairness evaluation (not as input). We form realistic vertical silos: Bank A (account/credit history), Employer/Payroll (employment/income), and an optional Bureau/Platform view (savings/collateral/third-party debtor), and evaluate under IID and non-IID splits.  \textbf{(2) Healthcare (UCI Heart Disease):} We use UCI Heart Disease \cite{janosi1988uci} with a binary disease label. The sensitive attribute is sex, treated as protected and omitted from inputs. We create three vertical views: Hospital A (clinical measurements/tests), Clinic B (demographic and structural risk factors), and Payer (remaining lab/administrative indicators), yielding clinically plausible mediators and proxies.  \textbf{(3) Criminal Justice (COMPAS Cox violent subset):} We use a COMPAS-style subset \cite{dressel2018accuracy} with a binary violent recidivism outcome. The sensitive attribute is race (African American vs.\ other), used only for fairness evaluation and excluded from inputs. Features are partitioned into Court system (charges/history), Community services (socioeconomic/supervision), and Agency/Platform (demographic screening), supporting both approximately IID and race-shifted non-IID splits.

We additionally evaluate SCC-VFL on a larger multi-group sensitive-attribute benchmark using \textbf{ACS Folktables} \cite{ding2021retiring}, with race (RAC1P) as a multi-category sensitive attribute and Income/Mobility as binary targets. Results under both IID and protected-attribute shift (Non-IID) settings are reported in Appendix~\ref{app:folktables}.

\subsection{Evaluation Metrics}
We evaluate each model using a fixed set of complementary metrics: \textbf{(1) Accuracy} (\%) $\uparrow$, standard predictive performance; \textbf{(2) LogLoss} $\downarrow$, cross-entropy over predicted probabilities; \textbf{(3) Selective Consistency Gap (SCG)} $\downarrow$, the mean $\|f(z)-f(z^{cf})\|$ over logits, capturing per-instance counterfactual stability; \textbf{(4) Flip Rate (FR, \%)} $\downarrow$, the fraction of samples whose predicted label changes under counterfactual edits (e.g., FR$=0.3\%$ corresponds to $\approx$3 flips per 1000 individuals); \textbf{(5) Attribute Inference Attack Success Rate (AIA SR, \%)} $\downarrow$, the post-hoc attacker success rate for predicting $s$ from latent representations under increasing attack strength; and \textbf{(6) Subspace-PGD Attack Success Rate (PGD SR, \%)} $\downarrow$, the success rate of an $\ell_\infty$-bounded adversary that perturbs mediator coordinates to flip predictions as $\epsilon$ grows. Together, these quantify utility (Accuracy, LogLoss), fairness (SCG, FR), and privacy leakage and robustness (attack success rates).

\subsection{Existing Baselines for Performance Comparison with SCC-VFL} 
We compare SCC-VFL against four baselines: \textbf{(1) No-mask adversarial debiasing (Adv-NoMask)} \citep{hong2021federated}, which applies a server-level adversary to remove $s$-signal from fused activations without counterfactual reasoning; \textbf{(2) Uniform counterfactuals (Uniform-CF)} \citep{garg2019counterfactual}, which uses a na\"{\i}ve generator that edits all coordinates indiscriminately; \textbf{(3) Policy-blind mask (Policy-blind Mask)} \citep{wang2022improving}, which performs descendant discovery without an explicit proxy partition; and \textbf{(4) Server-only consistency (Server-Consistency)} \citep{chen2024fair}, which directly penalizes prediction changes between original and ``flipped-$s$'' DP-sketch embeddings, without generator modules.

\subsection{Considered Adversaries} 
We stress-test all models with two adversarial probes: \textbf{(1) Attribute Inference Attack (AIA)} \citep{melis2019exploiting}, a post-hoc attacker that trains a separate classifier on frozen latent representations to recover $s$ under increasing attacker training budgets; and \textbf{(2) Subspace-constrained Projected Gradient Descent (Subspace-PGD)} \citep{pfrommer2023projected}, a gradient-based attack that perturbs only coordinates in a designated sensitive subspace within an $\ell_\infty$ ball of radius $\epsilon$ to flip the predicted label. Complete configurations and additional results for AIA and the Subspace-PGD variant are reported in Appendix~\ref{app:privacy_robustness}.

\subsection{Implementation Details}
All models are implemented in PyTorch. Each party encoder is a two-layer MLP (see Appendix~\ref{app:models-hparams}) with ReLU and dropout $d\in[0.05,0.10]$, and the server applies a linear prediction head (output dimension equals the number of classes). SCC-VFL uses a conditional VAE-style masked generator that perturbs only mediator coordinates, with counterfactual scale $\gamma\in[0.20,0.25]$; baselines use deterministic masked generators applied either to mediators or to all features. Mediators are selected via a DP-motivated discovery score that combines $\lvert \Delta_j^{(p)} \rvert$ and $\widehat{\mathrm{HSIC}}_j^{(p)}$ (Section~\ref{subsec:mask}); thresholds are set by percentile ranking, retaining the top $\rho_M=0.60$ coordinates as mediators $M^{(p)}$ and defining proxies $P^{(p)}$ as the top $\rho_P=0.50$ fraction within $M^{(p)}$. Training uses AdamW with learning rate $0.005$--$0.015$ and optional weight decay $5\times10^{-4}$ for up to 80 epochs for baselines and 150--300 epochs for SCC-VFL, with early stopping on a composite validation objective combining log loss, SCG, and FR. SCC-VFL logits are calibrated via temperature scaling on a held-out validation set (see Appendix~\ref{app:sensitivity}). All metrics are averaged over 30 random seeds and reported under both IID and non-IID client splits.

\subsection{Experimental Design}
We conduct three experiments: \textbf{(E1)} utility--fairness comparisons across all datasets using the core metrics; \textbf{(E2)} attack evaluations that measure sensitive leakage under party-side and server-side probes; and \textbf{(E3)} ablations that isolate the effects of mask discovery, masked generation, and the server-side consistency penalty on utility and stability. Together, these experiments show that SCC-VFL delivers counterfactual stability, improved privacy protection, and strong predictive performance across diverse VFL settings; (E2) and (E3) are conducted on the German Credit dataset under the IID split.

\begin{table*}[t]
  \caption{Performance comparison on UCI German Credit, UCI Heart Disease, and COMPAS Cox datasets under IID and Non-IID splits. Best values per block are highlighted in bold.}
  \label{tab:vfl_tabular_results}
  \centering
  \small
  \setlength{\tabcolsep}{5pt}
  \renewcommand{\arraystretch}{1.15}
  \begin{tabular}{lllcccc}
    \toprule
    Dataset & Split & Method
    & Acc $\uparrow$
    & LogLoss $\downarrow$
    & SCG $\downarrow$
    & FR (\%) $\downarrow$ \\
    \midrule
    \multirow{10}{*}{German Credit}
      & \multirow{5}{*}{IID}
        & Adv-NoMask
        & \meanstd{0.7176}{0.0217}
        & \meanstd{0.7131}{0.0666}
        & \meanstd{0.0561}{0.0076}
        & \meanstd{0.97}{0.35} \\
      &  & Uniform-CF
        & \meanstd{0.7311}{0.0197}
        & \meanstd{0.6305}{0.0423}
        & \meanstd{0.0891}{0.0094}
        & \meanstd{1.56}{0.77} \\
      &  & Policy-blind Mask
        & \meanstd{0.7322}{0.0214}
        & \meanstd{0.6114}{0.0433}
        & \meanstd{0.0532}{0.0053}
        & \meanstd{1.30}{0.75} \\
      &  & Server-Consistency
        & \textbf{\meanstd{0.7351}{0.0221}}
        & \meanstd{0.5747}{0.0409}
        & \meanstd{0.0218}{0.0029}
        & \meanstd{0.68}{0.48} \\
      &  & SCC-VFL (ours)
        & \meanstd{0.7224}{0.0244}
        & \textbf{\meanstd{0.5676}{0.0277}}
        & \textbf{\meanstd{0.0031}{0.0036}}
        & \textbf{\meanstd{0.09}{0.15}} \\
    \arrayrulecolor{gray!60}
    \cmidrule(lr){2-7}
    \arrayrulecolor{black}
      & \multirow{5}{*}{Non-IID}
        & Adv-NoMask
        & \meanstd{0.7200}{0.0291}
        & \meanstd{0.7143}{0.0676}
        & \meanstd{0.5688}{0.0306}
        & \meanstd{8.22}{1.75} \\
      &  & Uniform-CF
        & \meanstd{0.7293}{0.0229}
        & \meanstd{0.6376}{0.0510}
        & \meanstd{0.3652}{0.0297}
        & \meanstd{7.00}{1.54} \\
      &  & Policy-blind Mask
        & \textbf{\meanstd{0.7324}{0.0242}}
        & \meanstd{0.6155}{0.0454}
        & \meanstd{0.3185}{0.0270}
        & \meanstd{6.24}{1.35} \\
      &  & Server-Consistency
        & \meanstd{0.7300}{0.0230}
        & \meanstd{0.5904}{0.0460}
        & \meanstd{0.3335}{0.0273}
        & \meanstd{6.90}{1.39} \\
      &  & SCC-VFL (ours)
        & \meanstd{0.7243}{0.0246}
        & \textbf{\meanstd{0.5716}{0.0259}}
        & \textbf{\meanstd{0.0036}{0.0041}}
        & \textbf{\meanstd{0.12}{0.20}} \\
    \midrule
    \multirow{10}{*}{UCI Heart}
      & \multirow{5}{*}{IID}
        & Adv-NoMask
        & \meanstd{0.9712}{0.0053}
        & \textbf{\meanstd{0.0907}{0.0155}}
        & \meanstd{1.9882}{0.4007}
        & \meanstd{12.60}{3.01} \\
      &  & Uniform-CF
        & \meanstd{0.9530}{0.0080}
        & \meanstd{0.1806}{0.0105}
        & \meanstd{0.6694}{0.1049}
        & \meanstd{10.52}{2.61} \\
      &  & Policy-blind Mask
        & \meanstd{0.9488}{0.0067}
        & \meanstd{0.1941}{0.0083}
        & \meanstd{0.6067}{0.1171}
        & \meanstd{10.45}{2.88} \\
      &  & Server-Consistency
        & \meanstd{0.8844}{0.0211}
        & \meanstd{0.3442}{0.0165}
        & \meanstd{0.2869}{0.0319}
        & \meanstd{9.13}{1.48} \\
      &  & SCC-VFL (ours)
        & \textbf{\meanstd{0.9820}{0.0077}}
        & \meanstd{0.0910}{0.0136}
        & \textbf{\meanstd{0.2503}{0.0394}}
        & \textbf{\meanstd{2.11}{0.86}} \\
    \arrayrulecolor{gray!60}
    \cmidrule(lr){2-7}
    \arrayrulecolor{black}
      & \multirow{5}{*}{Non-IID}
        & Adv-NoMask
        & \meanstd{0.8538}{0.0386}
        & \meanstd{0.5798}{0.0605}
        & \meanstd{0.4646}{0.1562}
        & \meanstd{15.21}{5.57} \\
      &  & Uniform-CF
        & \meanstd{0.8557}{0.0340}
        & \meanstd{0.5804}{0.0418}
        & \meanstd{0.0935}{0.0182}
        & \meanstd{2.89}{1.46} \\
      &  & Policy-blind Mask
        & \meanstd{0.8500}{0.0390}
        & \meanstd{0.5883}{0.0489}
        & \meanstd{0.0248}{0.0064}
        & \meanstd{0.83}{0.67} \\
      &  & Server-Consistency
        & \meanstd{0.8145}{0.0432}
        & \meanstd{0.6282}{0.0555}
        & \meanstd{0.1840}{0.0522}
        & \meanstd{7.86}{2.95} \\
      &  & SCC-VFL (ours)
        & \textbf{\meanstd{0.9492}{0.0153}}
        & \textbf{\meanstd{0.4594}{0.0497}}
        & \textbf{\meanstd{0.0092}{0.0058}}
        & \textbf{\meanstd{0.05}{0.19}} \\
    \midrule
    \multirow{10}{*}{COMPAS Cox}
      & \multirow{5}{*}{IID}
        & Adv-NoMask
        & \textbf{\meanstd{0.9710}{0.0021}}
        & \textbf{\meanstd{0.0499}{0.0021}}
        & \meanstd{1.0357}{0.2266}
        & \meanstd{0.66}{0.46} \\
      &  & Uniform-CF
        & \meanstd{0.9705}{0.0034}
        & \meanstd{0.0548}{0.0016}
        & \meanstd{0.3839}{0.0747}
        & \meanstd{0.56}{0.57} \\
      &  & Policy-blind Mask
        & \meanstd{0.9705}{0.0035}
        & \meanstd{0.0548}{0.0021}
        & \meanstd{0.3901}{0.0812}
        & \meanstd{0.55}{0.55} \\
      &  & Server-Consistency
        & \meanstd{0.9706}{0.0034}
        & \meanstd{0.0553}{0.0021}
        & \meanstd{0.3191}{0.0360}
        & \meanstd{0.63}{0.64} \\
      &  & SCC-VFL (ours)
        & \meanstd{0.9697}{0.0042}
        & \meanstd{0.0514}{0.0024}
        & \textbf{\meanstd{0.0266}{0.0049}}
        & \textbf{\meanstd{0.02}{0.03}} \\
    \arrayrulecolor{gray!60}
    \cmidrule(lr){2-7}
    \arrayrulecolor{black}
      & \multirow{5}{*}{Non-IID}
        & Adv-NoMask
        & \meanstd{0.9758}{0.0015}
        & \textbf{\meanstd{0.0428}{0.0017}}
        & \meanstd{1.0712}{0.3014}
        & \meanstd{0.57}{0.34} \\
      &  & Uniform-CF
        & \meanstd{0.9750}{0.0047}
        & \meanstd{0.0485}{0.0027}
        & \meanstd{0.3812}{0.0694}
        & \meanstd{0.41}{0.34} \\
      &  & Policy-blind Mask
        & \meanstd{0.9755}{0.0032}
        & \meanstd{0.0485}{0.0018}
        & \meanstd{0.3751}{0.0793}
        & \meanstd{0.34}{0.28} \\
      &  & Server-Consistency
        & \textbf{\meanstd{0.9759}{0.0014}}
        & \meanstd{0.0489}{0.0019}
        & \meanstd{0.3238}{0.0441}
        & \meanstd{0.34}{0.28} \\
      &  & SCC-VFL (ours)
        & \meanstd{0.9721}{0.0065}
        & \meanstd{0.0453}{0.0035}
        & \textbf{\meanstd{0.0255}{0.0044}}
        & \textbf{\meanstd{0.02}{0.03}} \\
    \bottomrule
  \end{tabular}
\end{table*}


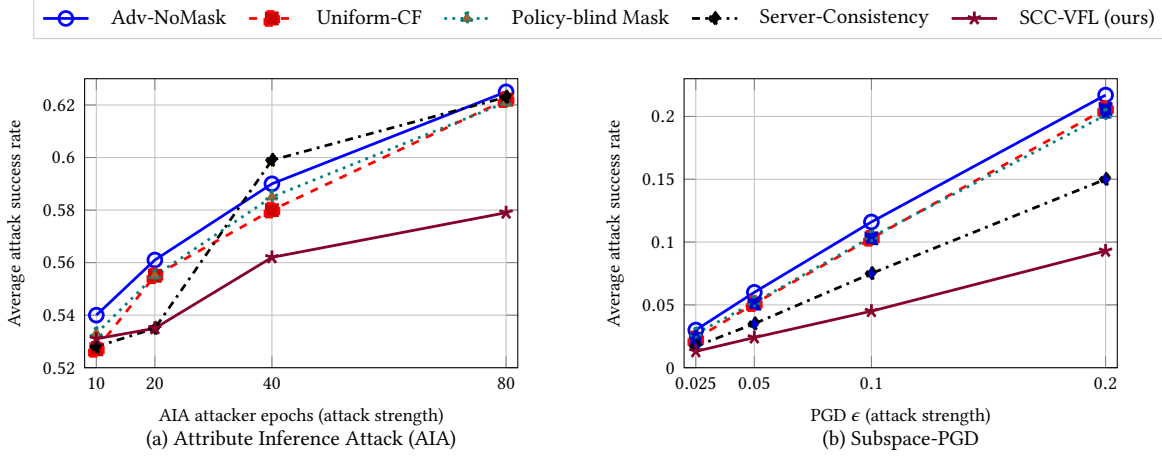
\begin{figure}[t]
\centering
\begin{tikzpicture}

\begin{groupplot}[
  group style={group size=2 by 1, horizontal sep=22mm},
  width=0.46\linewidth,          
  height=0.34\linewidth,         
  grid=both,
  grid style={line width=.1pt, draw=black!15},
  major grid style={line width=.2pt, draw=black!25},
  tick label style={font=\scriptsize},
  label style={font=\scriptsize},
  every axis plot/.append style={
    line width=1.05pt,
    mark options={mark size=2.6pt}
  }
]

\nextgroupplot[
  xlabel={AIA attacker epochs (attack strength)},
  ylabel={Average attack success rate},
  xmin=8, xmax=82,
  ymin=0.52, ymax=0.63,
  xtick={10,20,40,80},
  ytick={0.52,0.54,0.56,0.58,0.60,0.62},
  legend to name=attacklegend,
  legend columns=5,
  legend style={
    font=\footnotesize, 
    draw=black!20,
    fill=white,
    cells={anchor=west},
    /tikz/column sep=6pt
  }
]

\addplot+[color=blue, solid, mark=o]
  coordinates {(10,0.540) (20,0.561) (40,0.590) (80,0.625)};
\addlegendentry{Adv-NoMask}

\addplot+[color=red, dashed, mark=square*]
  coordinates {(10,0.527) (20,0.555) (40,0.580) (80,0.622)};
\addlegendentry{Uniform-CF}

\addplot+[color=teal, dotted, mark=triangle*]
  coordinates {(10,0.533) (20,0.555) (40,0.585) (80,0.621)};
\addlegendentry{Policy-blind Mask}

\addplot+[color=black, dashdotted, mark=diamond*]
  coordinates {(10,0.528) (20,0.535) (40,0.599) (80,0.623)};
\addlegendentry{Server-Consistency}

\addplot+[color=purple!70!black, solid, mark=star]
  coordinates {(10,0.531) (20,0.535) (40,0.562) (80,0.579)};
\addlegendentry{SCC-VFL (ours)}

\nextgroupplot[
  xlabel={PGD $\epsilon$ (attack strength)},
  ylabel={Average attack success rate},
  xmin=0.02, xmax=0.205,
  ymin=0.00, ymax=0.23,
  scaled x ticks=false,
  xticklabel style={/pgf/number format/fixed, /pgf/number format/precision=3},
  xtick={0.025,0.050,0.100,0.200},
  scaled y ticks=false,
  ytick={0,0.05,0.10,0.15,0.20,0.25},
  yticklabel style={/pgf/number format/fixed, /pgf/number format/precision=2},
]

\addplot+[color=blue, solid, mark=o, forget plot]
  coordinates {(0.025,0.030) (0.050,0.060) (0.100,0.116) (0.200,0.217)};

\addplot+[color=red, dashed, mark=square*, forget plot]
  coordinates {(0.025,0.023) (0.050,0.051) (0.100,0.103) (0.200,0.206)}; 

\addplot+[color=teal, dotted, mark=triangle*, forget plot]
  coordinates {(0.025,0.027) (0.050,0.052) (0.100,0.104) (0.200,0.201)};

\addplot+[color=black, dashdotted, mark=diamond*, forget plot]
  coordinates {(0.025,0.017) (0.050,0.035) (0.100,0.075) (0.200,0.150)};

\addplot+[color=purple!70!black, solid, mark=star, forget plot]
  coordinates {(0.025,0.013) (0.050,0.024) (0.100,0.045) (0.200,0.093)};

\end{groupplot}

\node[anchor=south] at ($ (group c1r1.north)!0.5!(group c2r1.north) + (0,4mm) $)
{\pgfplotslegendfromname{attacklegend}};

\node[anchor=north, font=\footnotesize] at ($ (group c1r1.south) + (0,-7mm) $)
{(a) Attribute Inference Attack (AIA)};
\node[anchor=north, font=\footnotesize] at ($ (group c2r1.south) + (0,-7mm) $)
{(b) Subspace-PGD};

\end{tikzpicture}
\caption{Average attack success rate for AIA and Subspace-PGD across attack strengths (30 seeds).}
\label{fig:attack_plots}
\end{figure}

\section{Experimental Results \& Analyses}

\subsection{Utility--Fairness Trade-offs Across Datasets}
We first compare \textsc{SCC-VFL} and all baselines on the joint utility--fairness metrics across the three datasets under both IID and non-IID partitions (Table~\ref{tab:vfl_tabular_results}). Across domains, \textsc{SCC-VFL} either matches or is very close to the best baseline in Accuracy and LogLoss while sharply improving Selective Consistency Gap (SCG) and Flip Rate (FR). On German Credit, \textsc{SCC-VFL} slightly trails Server-Consistency in Accuracy by about one percentage point but attains the lowest LogLoss and reduces SCG and FR by roughly an order of magnitude in both IID and non-IID settings. On UCI Heart, \textsc{SCC-VFL} achieves the highest Accuracy in both splits and essentially ties the best LogLoss, while cutting SCG and FR from double digits to around \(0.25\) / \(2.1\) in IID and almost zero in non-IID. On COMPAS Cox, \textsc{SCC-VFL} maintains Accuracy within \(0.3\) percentage points of the strongest baselines, with only a minor LogLoss gap, yet collapses SCG to about \(0.026\) and FR to \(0.02\) compared to \(0.3\)--\(1.0\) SCG and \(0.3\)--\(0.7\) FR for all baselines. non-IID partitions intensify these gaps: baselines often suffer large SCG and FR spikes when distributions shift, whereas \textsc{SCC-VFL} keeps SCG in the \(0.003\)--\(0.026\) range and FR near zero while staying competitive in utility.

These patterns follow from how \textsc{SCC-VFL} enforces counterfactual stability. Adv-NoMask operates only at the fused representation, so it can suppress some global \(s\)-signal but lacks instance-level control over which features move, leaving many per-example flips and high SCG. Uniform-CF perturbs all coordinates, often preserving or slightly improving utility but introducing off-support edits that inflate SCG and FR, especially under non-IID splits. Policy-blind Mask performs better by focusing edits on descendants, yet treating all descendants as acceptable mediators lets proxies transmit residual sensitive signal and limits stability gains. Server-Consistency constrains predictions under ``flipped-\(s\)'' embeddings but has no explicit generator or identity constraint, so it cannot guarantee that only permissible mediators move or that non-descendants stay fixed. In contrast, \textsc{SCC-VFL} combines selective mask discovery, mediator-only cVAE edits, and a per-example SCC loss at the server, which jointly preserves on-support edits, freezes non-descendants, and collapses proxy influence. This architecture yields models that keep task performance high while producing much lower SCG and FR across datasets and under both IID and non-IID VFL regimes.

\subsection{Privacy Leakage and Robustness Analysis}
We next evaluate privacy leakage and robustness using the Attribute Inference Attack (AIA) and Subspace-constrained Projected Gradient Descent (Subspace-PGD). As Fig.~\ref{fig:attack_plots} shows, as attacker strength increases, baselines cluster around \(62\%\) AIA success at 80 epochs, whereas \textsc{SCC-VFL} stays at \(57.9\%\), a \(4\)--\(5\) point reduction in balanced attack accuracy. For Subspace-PGD, which perturbs only mediator coordinates, \textsc{SCC-VFL} consistently attains the lowest success rate, from \(1.0\%\) at \(\epsilon=0.02\) to \(9.3\%\) at \(\epsilon=0.20\), about half of Adv-NoMask (\(21.7\%\)) and below all other baselines. These gains stem from combining selective mask discovery (to isolate mediators and proxies), masked generation (editing only mediators while enforcing identity on non-descendants and suppressing proxy leakage), and an SCC loss that flattens the decision boundary along mediator directions, reducing usable \(s\)-signal and strengthening resistance to mediator-space attacks relative to Adv-NoMask, Uniform-CF, Policy-blind Mask, and Server-Consistency.

\subsection{Ablation Analysis }
Table~\ref{tab:scc_vfl_ablation_german} studies the contribution of each component of \textsc{SCC-VFL}. The full model attains the best overall trade-off, with the lowest SCG (\(0.0045\)) and FR (\(0.11\%\)) while keeping Accuracy and LogLoss on par with or slightly better than the ablations. Removing mask discovery and treating all features as mediators yields marginally lower LogLoss but roughly doubles SCG and more than triples FR, showing that selective masks help target edits and avoid unnecessary flips. Removing the generator (leaving only server-side consistency) substantially degrades stability, with SCG increasing to \(0.0233\) and FR to \(0.94\%\), indicating that explicit on-support mediator counterfactuals are important for effective consistency enforcement. Finally, dropping the consistency term while keeping masks and the generator retains good utility but leaves SCG and FR significantly higher than the full model, confirming that the per-example SCC loss at the server is necessary to translate selective edits into stable predictions. We also report runtime and communication overhead measurements in Appendix~\ref{app:runtime-comm}.

\begin{table}[t]
  \caption{Ablation of SCC-VFL components on German Credit (IID split).}
  \label{tab:scc_vfl_ablation_german}
  \centering
    \begin{tabular}{@{}lcccc@{}}
      \toprule
      Method & Acc $\uparrow$ & LogLoss $\downarrow$ & SCG $\downarrow$ & FR (\%) $\downarrow$ \\
      \midrule
      SCC-VFL (full) &
        \meanstd{0.7186}{0.022} &
        \meanstd{0.5737}{0.027} &
        \textbf{\meanstd{0.0045}{0.004}} &
        \textbf{\meanstd{0.11}{0.180}} \\
      w/o mask discovery &
        \textbf{\meanstd{0.7220}{0.024}} &
        \textbf{\meanstd{0.5713}{0.026}} &
        \meanstd{0.0094}{0.016} &
        \meanstd{0.36}{0.60} \\
      w/o generator &
        \meanstd{0.7089}{0.017} &
        \meanstd{0.5761}{0.027} &
        \meanstd{0.0233}{0.005} &
        \meanstd{0.94}{0.54} \\
      w/o consistency &
        \meanstd{0.7174}{0.022} &
        \meanstd{0.5747}{0.027} &
        \meanstd{0.0076}{0.000} &
        \meanstd{0.63}{0.100} \\
      \bottomrule
    \end{tabular}%
\end{table}

\subsection{Discussion \& Implications}
\label{sec:discussion}

\textbf{What SCC-VFL does and does not guarantee.}
Our results show that SCC-VFL substantially reduces prediction instability under counterfactual interventions on the sensitive attribute. This is a \emph{consistency} property: for a given individual, predictions remain stable when the protected attribute is counterfactually changed, non-descendant features are held fixed, and only policy-permitted mediators are allowed to vary. Viewed this way, mediator edits function as recourse-style interventions rather than arbitrary perturbations~\cite{ustun2019actionable, wachter2017counterfactual}, and our guarantees inherit standard caveats about counterfactual-based fairness~\cite{barocas2020hidden, kasirzadeh2021use}.

Consistency, however, does not imply justice~\cite{selbst2019fairness}. A system can be consistently discriminatory if the underlying policy specification encodes biased or unjustified judgments about which causal pathways are deemed permissible~\cite{selbst2019fairness, weerts2023algorithmic}. SCC-VFL enforces the provided policy model rather than correcting it. Moreover, counterfactual stability captures only one dimension of fairness and does not ensure group-level parity properties such as equalized odds or demographic parity, nor does it resolve intersectional concerns. Deployments should therefore complement SCC-VFL with group-level audits and substantive review of whether the policy model itself reflects defensible values.

\textbf{Individual stability as an auditable fairness target.}
The results position individual-level counterfactual stability as a practically auditable fairness target for VFL. Metrics such as Selective Consistency Gap (SCG) and Flip Rate (FR) directly quantify whether a person’s prediction changes under an intervention on the protected attribute when non-descendants are held fixed. Across datasets, SCC-VFL achieves large reductions in SCG and FR while maintaining competitive Accuracy and LogLoss, indicating that it removes arbitrary per-example sensitivity rather than trading utility for coarse group-level constraints~\cite{binns2020apparent}. This aligns with FAccT goals of accountability and contestability, since stability can be assessed at the individual level and aggregated for monitoring.

\textbf{Robustness under heterogeneity and shift.}
The non-IID results highlight how institutional heterogeneity and distribution shift can amplify proxy pathways and off-support edits, leading baseline methods to exhibit sharp stability degradation. SCC-VFL remains more stable by restricting edits to mediator coordinates, enforcing identity on non-descendants, and applying a per-example consistency penalty that flattens the decision boundary along the intended intervention. The observed reductions in attribute inference and Subspace-PGD attack success suggest reduced usable sensitive signal in shared representations under the evaluated threats, though these results should be interpreted as empirical robustness rather than formal privacy guarantees. A key implication is that policy choices embedded in the mask should be documented and audited, including shift-aware reporting of SCG and FR, to ensure that enforced intervention semantics align with domain requirements.

\subsection{Limitations and Broader Considerations}
\label{sec:limitations}

\noindent \textbf{Technical scope.}
We evaluate SCC-VFL on tabular VFL benchmarks with discrete protected attributes and a fixed policy specification. Extending the approach to high-dimensional modalities (text, images), continuous protected attributes, and settings with frequent policy updates requires additional design for (i) scalable mask discovery and (ii) valid, on-support counterfactual generation in richer feature spaces. The proposed mask discovery is graph-free and relies on statistical signals in order to support the enforcement of a declared policy rather than to recover ground-truth causal structure.

\noindent \textbf{Policy dependence and governance.}
SCC-VFL enforces counterfactual stability relative to the chosen $N/M/P$ specification; it does not determine which pathways \emph{are normatively permissible}. If the specification is incomplete or outdated, the model can be stable while still misaligned with institutional or legal intent. For deployment, we recommend documenting the policy rationale, tracking mask versions, and pairing SCC metrics (e.g., SCG/FR) with routine outcome monitoring, including intersectional slices when feasible, following model-reporting practices~\cite{mitchell2019model}.

\noindent \textbf{Privacy and accountability boundaries.} 
The formal DP guarantee applies only to the sketch used for mask discovery (Proposition~\ref{prop:dp_sketch}), not to representations, gradients, or model parameters; composition across mask refreshes is capped but not tracked via a formal composition accountant (e.g., R{\'e}nyi DP). Beyond the sketch, privacy evidence is empirical and does not rule out inference under stronger threat models. We position SCC-VFL as an accountability mechanism to be used alongside established privacy protections, audits, and external review.

\section{Conclusion and Future Work}

This paper introduced \textsc{SCC-VFL}, a server-centric framework for enforcing individual-level counterfactual stability in vertical federated learning while preserving predictive performance. The approach integrates three components: selective mask discovery that partitions each party’s features into non-descendants, mediators, and proxies; masked counterfactual generation that edits only policy-permitted mediators while preserving identity on non-descendants and suppressing proxy leakage; and a server-side selective counterfactual consistency loss that penalizes prediction changes under on-support mediator interventions. Across banking, healthcare, and criminal justice datasets, \textsc{SCC-VFL} achieves strong predictive utility while substantially reducing the selective consistency gap, flip rate, and adversarial attack success under both attribute inference and subspace-constrained PGD, in IID and non-IID vertical settings.

Future work includes extending \textsc{SCC-VFL} along methodological, privacy, and deployment dimensions. An immediate direction is supporting higher-dimensional modalities such as text or imaging, and more complex vertical consortia with many parties, partial entity overlap, and asynchronous participation. On the privacy side, integrating formal differential privacy accounting and secure computation primitives into the mask discovery stage would provide end-to-end privacy guarantees beyond the current functional protections. Another avenue is adaptive and multi-attribute mask discovery that can track distribution shift and handle multiple protected attributes simultaneously. Finally, the masked generators developed for counterfactual stability could also support actionable recourse and auditing~\cite{ustun2019actionable}, where policy-compliant mediator edits are exposed to stakeholders as interpretable explanations or intervention recommendations.

\section*{Code and data availability}
The implementation is available at \url{https://github.com/dawoodwasif/SCC-VFL}. The datasets used in this work are publicly available from their official sources; links and preprocessing scripts are provided in the repository.

\section*{Generative AI Usage}
We used generative models in two limited and transparent ways. First, a generative model is used as a \emph{method component}: the masked conditional VAE that produces mediator-only counterfactual edits within the proposed SCC-VFL framework. Second, standard generative tools were used during manuscript preparation for formatting assistance and language polishing. Generative models were not used to generate scientific hypotheses, select evaluation metrics, design experiments, interpret results, or draw conclusions. No external generative system was used to create, label, or augment the experimental datasets, and no model outputs were used as ground-truth labels for evaluation. All reported quantitative results, tables, and figures are produced by our implemented training and evaluation code on the stated datasets and splits.

\section*{Ethical Considerations}
This work studies fairness and privacy risks in vertical federated learning, where disjoint feature holders collaborate without sharing raw features. The protected attribute $s$ is treated as sensitive: it is not broadcast to parties and is used only to define and evaluate the counterfactual intervention underlying Selective Consistency Gap (SCG) and Flip Rate (FR). Our approach encodes an explicit normative choice about which descendants of $s$ are permissible mediators versus impermissible proxies. These choices can shape downstream incentives and institutional behavior and should therefore be defined using domain policy, legal guidance, and stakeholder input. We emphasize that deployments should document these choices and audit resulting models using both utility metrics and stability metrics such as SCG and FR across relevant subgroups and distribution shifts.

\section*{Adverse Impacts}
A stability objective can be misused to justify suppressing legitimate signal, to enforce overly rigid decisions, or to mask discrimination by optimizing a narrow metric. Errors in mediator or proxy specification can unintentionally amplify harm by freezing features that should change or allowing pathways that should have been blocked. While our privacy results show reduced post hoc leakage under the evaluated attribute inference attacks and increased robustness under Subspace-PGD, the method does not provide a formal privacy guarantee for training data, representations, gradients, or model parameters. Adversaries outside our threat model may still infer sensitive information. As with many fairness interventions, the risks of mis-specification or misuse are most likely to be borne by protected or marginalized groups, underscoring the need for external accountability. To mitigate these risks, we recommend documenting intervention semantics and policy specifications, reporting SCG and FR alongside standard group fairness and utility metrics under non-IID and temporal shifts, and monitoring deployed systems with clear review and rollback processes.

\begin{acks}
This work was supported in part by the U.S. Army Research Office (ARO) under Award No.~W911NF-24-2-0241 and in part by the U.S. National Science Foundation (NSF) under Grant No.~2416728. The views and conclusions contained in this paper are those of the authors and should not be interpreted as representing the official policies, either expressed or implied, of the U.S. Army Research Office, the U.S. Government, the National Science Foundation, or any other funding agency. Any opinions, findings, conclusions, or recommendations expressed in this material are those of the authors and do not necessarily reflect the views of the funding agencies.
\end{acks}

\bibliographystyle{ACM-Reference-Format}
\bibliography{references}

\newpage
\begin{center}
  \LARGE \textbf{APPENDICES}
\end{center}
\vspace{1em}

\appendix

\counterwithin{figure}{section}
\counterwithin{table}{section}
\renewcommand{\thefigure}{\thesection.\arabic{figure}}
\renewcommand{\thetable}{\thesection.\arabic{table}}


\section{Datasets and Client Partitions}
\label{app:data-partitions}

We use three tabular benchmarks with a held-out protected attribute (used only for fairness/privacy evaluation). Dataset-level details are in Table~\ref{tab:data-summary}; vertical partitions are in Table~\ref{tab:client-partitions}.

\subsection{Overview of datasets}
Preprocessing and protected-attribute operationalization choices below are documented in the spirit of dataset transparency practices~\cite{gebru2021datasheets}.

\textbf{German Credit (Banking).} Binary credit risk (bad vs good); protected attribute: age binarised at 25; age excluded from inputs.

\noindent\textbf{UCI Heart Disease (Healthcare).} Binary heart disease presence; protected attribute: sex; sex excluded from inputs; remaining features standardised.

\noindent\textbf{COMPAS Cox violent subset (Criminal justice).} Binary violent recidivism within horizon; protected attribute: race (African American vs other); race excluded from inputs; categorical encoded, numerics imputed then scaled.

\subsection{Client partitions and VFL views}

Each dataset is split into three feature-holding clients plus a coordinating server; clients send only encodings. Concrete per-dataset client feature blocks are in Table~\ref{tab:client-partitions}.

\begin{table}[H]
  \caption{Dataset summary. Here $n$ is the number of samples and $d$ is the number of numerical features after preprocessing. Protected attributes are used only for fairness and privacy evaluation, and are never provided as inputs to the prediction models.}
  \label{tab:data-summary}
  \centering
  \small
  \renewcommand{\arraystretch}{1.15}
  \resizebox{\linewidth}{!}{%
  \begin{tabular}{lcccccc}
    \toprule
    Dataset & Domain & $n$ & $d$ & Protected & Label & Splits \\
    \midrule
    German Credit & Banking & $\approx 1000$ & $\approx 20$ & Age ($< 25$ vs $\ge 25$) &
    Default (bad vs good) & IID and non-IID \\
    Heart Disease & Healthcare & $\approx 300$ & 13 & Sex (female vs male) &
    Presence of heart disease & IID and non-IID \\
    COMPAS Cox & Criminal justice & $\approx 6000$ & $\approx 20$ &
    Race (African American vs other) &
    Violent recidivism event & IID and non-IID \\
    \bottomrule
  \end{tabular}%
  }
\end{table}
\begin{table}[t]
  \caption{Vertical client partitions for each dataset. Each client holds a disjoint subset of feature groups and participates in the VFL protocol by sending local encodings to the server. The server coordinates training and aggregation but never observes raw features.}
  \label{tab:client-partitions}
  \centering
  \small
  \setlength{\tabcolsep}{4pt}
  \renewcommand{\arraystretch}{1.15}
  \begin{tabular}{llp{7.3cm}}
    \toprule
    Dataset & Client & Feature groups (examples) \\
    \midrule
    German Credit &
    Bank A &
    Account status, credit history, credit amount, credit duration, instalment rate \\
    &
    Employer / Payroll &
    Employment duration, job category, housing, number of dependents \\
    &
    Bureau / Platform &
    Savings accounts, collateral, other debtors or guarantors, foreign worker flag \\
    \midrule
    Heart Disease &
    Hospital A &
    Core clinical measures such as resting blood pressure, serum cholesterol, ST depression, maximum heart rate \\
    &
    Clinic B &
    Demographics and structural risk factors including age, chest pain type, vessel count, slope category \\
    &
    Payer &
    Fasting blood sugar, resting ECG findings, thallium scan result, simple administrative codes \\
    \midrule
    COMPAS Cox &
    Court system &
    Charge degree, offense type, prior counts, custody history, incarceration days \\
    &
    Community services &
    Employment status, education level, supervision and program participation indicators \\
    &
    Agency / Platform &
    Screening age, marital status, residence stability and other demographic records \\
    \bottomrule
  \end{tabular}
\end{table}

\subsection{Splitting strategies and non-IID settings}

We use stratified train/val/test splits (70/30 train/test; within train, 20\% validation). IID: stratified random by individuals. non-IID: shift protected-attribute distribution between train/test (race shift for COMPAS; sex shift for Heart Disease; analogous shift for German Credit).

\subsection{Protected-attribute operationalization and shift rationale}
\label{app:data-rationale}

\textbf{Protected attribute thresholds.}
We operationalize the protected attribute by a fixed, dataset-standard binarization to enable auditable group metrics while keeping the attribute out of model inputs (Table~\ref{tab:data-summary}). For German Credit we use age $<25$ vs.\ $\ge 25$ as an interpretable “younger vs.\ older” split; the SCC-VFL pipeline is unchanged under alternative fixed cutoffs because it only requires re-defining $s$ for evaluation and mask discovery.

\textbf{Non-IID shift realism.}
Our Non-IID setting is a targeted demographic shift: we alter the train/test distribution of the protected attribute (race for COMPAS, sex for Heart Disease, analogous for German) while keeping the feature space, labels, and client partitions unchanged. This models a common deployment mismatch where training data is collected under institution- or region-specific coverage, but evaluation occurs on a different demographic mix, stressing counterfactual stability under group shift.

\textbf{Partition and feature-family intuition.}
Client partitions (Table~\ref{tab:client-partitions}) are chosen to mimic realistic institutional boundaries and to keep the mediator/proxy discussion interpretable: German Credit concentrates socio-economic and employment-related variables in parties that plausibly mediate age effects; Heart partitions separate clinical measures from demographics and payer-coded variables; COMPAS partitions separate court records from community services and platform demographics. These choices make it transparent which feature families can enter the editable set $M$ and which can concentrate proxy signal $P$ under the same discovery score.

\section{Additional Benchmark}
\label{app:folktables}

We extend the experiments with a larger, multi-group sensitive-attribute benchmark using Folktables. We use the ACS 2018 1-Year Person survey (CA-only), subsample to 60k records, and follow the same evaluation protocol as the main paper: identical training pipeline and aggregation over 30 random seeds. We evaluate both IID and Non-IID settings.

\paragraph{Tasks and labels.}
We consider two Folktables prediction tasks: \textsc{Income} (binary label indicating whether income exceeds the task-defined threshold) and \textsc{Mobility} (binary mobility outcome as defined by Folktables).

\paragraph{Sensitive attribute.}
The sensitive attribute is race \texttt{RAC1P}, treated as multi-group and remapped to $s\in\{0,\dots,K-1\}$ (typically $K=9$ groups in this CA-only subset). As in the main paper, $s$ is used only for fairness/privacy evaluation and counterfactual construction, and is not provided as an input feature to the predictor.

\paragraph{Splits and shifts.}
IID uses stratified random splits by individuals. Non-IID uses a protected-attribute distribution shift between train and test induced by reweighting/sampling across race groups, mirroring the main paper’s shift-based protocol. For counterfactual evaluation with multi-group $s$, we sample an alternative group $s'\neq s$ per example and compute stability metrics under the corresponding intervention.

\paragraph{Metrics.}
We report utility (Accuracy, LogLoss), individual stability (SCG, FR), and group-fairness context metrics. Specifically, \emph{Demographic Parity (DP) difference} is the absolute difference (or, for multi-group, the max pairwise gap) in positive prediction rates across sensitive groups, and \emph{Equalized Odds (EO) gap} is the maximum absolute gap in TPR and FPR across groups.

\begin{table}[t]
\caption{Folktables ACS 2018 (CA-only, 60k) \textsc{Income} under IID splits, mean$\pm$std over 30 seeds. Best per column in bold.}
\label{tab:folktables_income_iid}
\centering
\small
\setlength{\tabcolsep}{4pt}
\renewcommand{\arraystretch}{1.15}
\begin{tabular}{lcccccc}
\toprule
Method & Acc.$\uparrow$ & LogLoss$\downarrow$ & SCG$\downarrow$ & FR$\downarrow$ & DP$\downarrow$ & EO$\downarrow$ \\
\midrule
Baseline     & $0.7954\pm0.0031$ & $0.4432\pm0.0035$ & $0.0261\pm0.0092$ & $0.7911\pm0.3142$ & $0.5822\pm0.1663$ & $0.6055\pm0.1888$ \\
Uniform CF   & $0.7951\pm0.0028$ & $0.4435\pm0.0031$ & $0.0214\pm0.0038$ & $0.5409\pm0.1263$ & $0.5786\pm0.1728$ & $0.5815\pm0.1773$ \\
Policy-blind & $0.7951\pm0.0032$ & $0.4441\pm0.0037$ & $\mathbf{0.0128\pm0.0016}$ & $0.3224\pm0.0538$ & $0.5826\pm0.1712$ & $0.5809\pm0.1756$ \\
Server-only  & $0.7953\pm0.0024$ & $0.4477\pm0.0031$ & $0.0213\pm0.0053$ & $0.7328\pm0.1637$ & $\mathbf{0.5759\pm0.1717}$ & $\mathbf{0.5679\pm0.1827}$ \\
SCC-VFL      & $\mathbf{0.7987\pm0.0033}$ & $\mathbf{0.4327\pm0.0044}$ & $0.0164\pm0.0118$ & $\mathbf{0.2046\pm0.1038}$ & $0.5771\pm0.1714$ & $0.5746\pm0.1883$ \\
\bottomrule
\end{tabular}
\end{table}

\begin{table}[t]
\caption{Folktables ACS 2018 (CA-only, 60k) \textsc{Mobility} under IID splits, mean$\pm$std over 30 seeds. Best per column in bold.}
\label{tab:folktables_mobility_iid}
\centering
\small
\setlength{\tabcolsep}{4pt}
\renewcommand{\arraystretch}{1.15}
\begin{tabular}{lcccccc}
\toprule
Method & Acc.$\uparrow$ & LogLoss$\downarrow$ & SCG$\downarrow$ & FR$\downarrow$ & DP$\downarrow$ & EO$\downarrow$ \\
\midrule
Baseline     & $0.7074\pm0.0072$ & $0.5738\pm0.0041$ & $0.0137\pm0.0040$ & $0.4631\pm0.1849$ & $0.1955\pm0.0704$ & $0.1336\pm0.0503$ \\
Uniform CF   & $0.7055\pm0.0092$ & $0.5751\pm0.0035$ & $0.0105\pm0.0026$ & $0.3539\pm0.1797$ & $0.1788\pm0.0827$ & $0.1235\pm0.0612$ \\
Policy-blind & $0.7034\pm0.0098$ & $0.5752\pm0.0045$ & $0.0084\pm0.0022$ & $0.2528\pm0.1472$ & $\mathbf{0.1505\pm0.0817}$ & $\mathbf{0.1020\pm0.0591}$ \\
Server-only  & $0.7022\pm0.0102$ & $0.5768\pm0.0040$ & $0.0114\pm0.0035$ & $0.3541\pm0.2614$ & $0.1624\pm0.1672$ & $0.1198\pm0.1702$ \\
SCC-VFL      & $\mathbf{0.7237\pm0.0044}$ & $\mathbf{0.5451\pm0.0030}$ & $\mathbf{0.0072\pm0.0199}$ & $\mathbf{0.0220\pm0.8754}$ & $0.1689\pm0.0433$ & $0.1104\pm0.0350$ \\
\bottomrule
\end{tabular}
\end{table}

\begin{table}[t]
\caption{Folktables ACS 2018 (CA-only, 60k) \textsc{Income} under Non-IID splits, mean$\pm$std over 30 seeds. Best per column in bold.}
\label{tab:folktables_income_noniid}
\centering
\small
\setlength{\tabcolsep}{4pt}
\renewcommand{\arraystretch}{1.15}
\begin{tabular}{lcccccc}
\toprule
Method & Acc.$\uparrow$ & LogLoss$\downarrow$ & SCG$\downarrow$ & FR$\downarrow$ & DP$\downarrow$ & EO$\downarrow$ \\
\midrule
Baseline     & $0.7337\pm0.0072$ & $0.5325\pm0.0085$ & $0.0325\pm0.0109$ & $1.1957\pm0.3764$ & $0.6041\pm0.1363$ & $0.5355\pm0.2910$ \\
Uniform CF   & $0.7327\pm0.0077$ & $0.5343\pm0.0093$ & $0.0206\pm0.0031$ & $0.6731\pm0.1176$ & $\mathbf{0.5968\pm0.1301}$ & $0.5338\pm0.2919$ \\
Policy-blind & $0.7353\pm0.0061$ & $0.5316\pm0.0069$ & $0.0150\pm0.0026$ & $0.4704\pm0.0990$ & $0.6036\pm0.1347$ & $0.5452\pm0.2804$ \\
Server-only  & $\mathbf{0.7378\pm0.0052}$ & $\mathbf{0.5267\pm0.0057}$ & $0.0248\pm0.0064$ & $1.0670\pm0.2953$ & $0.5994\pm0.1374$ & $0.5372\pm0.2842$ \\
SCC-VFL      & $0.7314\pm0.0112$ & $0.5273\pm0.0136$ & $\mathbf{0.0141\pm0.0066}$ & $\mathbf{0.3665\pm0.1614}$ & $0.6016\pm0.1294$ & $\mathbf{0.5020\pm0.2611}$ \\
\bottomrule
\end{tabular}
\end{table}

\begin{table}[t]
\caption{Folktables ACS 2018 (CA-only, 60k) \textsc{Mobility} under Non-IID splits, mean$\pm$std over 30 seeds. Best per column in bold.}
\label{tab:folktables_mobility_noniid}
\centering
\small
\setlength{\tabcolsep}{4pt}
\renewcommand{\arraystretch}{1.15}
\begin{tabular}{lcccccc}
\toprule
Method & Acc.$\uparrow$ & LogLoss$\downarrow$ & SCG$\downarrow$ & FR$\downarrow$ & DP$\downarrow$ & EO$\downarrow$ \\
\midrule
Baseline     & $0.6368\pm0.0083$ & $0.6613\pm0.0050$ & $0.0238\pm0.0064$ & $1.5080\pm0.4802$ & $0.4092\pm0.0285$ & $0.3132\pm0.0184$ \\
Uniform CF   & $0.6362\pm0.0087$ & $0.6602\pm0.0068$ & $0.0157\pm0.0037$ & $0.7915\pm0.1598$ & $0.4060\pm0.0214$ & $0.3147\pm0.0178$ \\
Policy-blind & $0.6345\pm0.0095$ & $0.6601\pm0.0073$ & $\mathbf{0.0121\pm0.0023}$ & $0.6357\pm0.1536$ & $0.4096\pm0.0227$ & $0.3175\pm0.0185$ \\
Server-only  & $\mathbf{0.6483\pm0.0079}$ & $\mathbf{0.6462\pm0.0063}$ & $0.0225\pm0.0087$ & $1.5256\pm0.6073$ & $\mathbf{0.3877\pm0.0323}$ & $0.2946\pm0.0395$ \\
SCC-VFL      & $0.6320\pm0.0184$ & $0.6542\pm0.0093$ & $0.0133\pm0.0280$ & $\mathbf{0.3567\pm2.1679}$ & $0.4024\pm0.0463$ & $\mathbf{0.2941\pm0.0350}$ \\
\bottomrule
\end{tabular}
\end{table}

Across both tasks, SCC-VFL attains the lowest flip rate in all four settings and remains competitive on utility. In IID, SCC-VFL achieves the best LogLoss in both tasks and the best Accuracy on \textsc{Mobility} (Table~\ref{tab:folktables_mobility_iid}), while substantially reducing decision instability (FR) relative to all baselines (Tables~\ref{tab:folktables_income_iid}--\ref{tab:folktables_mobility_iid}). Under Non-IID shifts, SCC-VFL preserves low FR and achieves the best EO in both tasks (Tables~\ref{tab:folktables_income_noniid}--\ref{tab:folktables_mobility_noniid}), indicating that the counterfactual stability objective remains well-controlled even under demographic distribution shift with multi-group $s$.

\section{Models and Hyperparameters}
\label{app:models-hparams}

All methods share the same backbone, optimiser schedule, and partitions (Table~\ref{tab:client-partitions}); differences are only in enabled components (Table~\ref{tab:hparams-method}). Shared hyperparameters are in Table~\ref{tab:hparams-common}.

\subsection{Common architecture and client--server protocol}

Client $c\in\{1,2,3\}$ holds $x_i^{(c)}\in\mathbb{R}^{d_c}$ and sends an encoding to the server:
\begin{equation}
h^{(c)}_i \;=\; h_{\phi}^{(c)}\!\left(x^{(c)}_i\right),
\qquad
z_i \;=\; \mathrm{Fuse}\!\left(h^{(1)}_i,\, h^{(2)}_i,\, h^{(3)}_i\right).
\label{eq:client_fuse}
\end{equation}
We implement this as a 2-layer MLP over concatenated $x_i=[x_i^{(1)},x_i^{(2)},x_i^{(3)}]$ with fixed index ranges:
\begin{equation}
h_i \;=\; \sigma\!\left(W_2\,\sigma\!\left(W_1 x_i + b_1\right) + b_2\right),
\label{eq:mlp_backbone}
\end{equation}
with ReLU $\sigma$. The server uses a linear classifier
\begin{equation}
f_\theta(z_i) \;=\; W_{\mathrm{cls}} h_i + b_{\mathrm{cls}},
\label{eq:server_head}
\end{equation}
for binary $y_i\in\{0,1\}$.

\subsection{Mediator and proxy discovery}
\label{app:mediator-proxy}

Server forms a DP sketch of the protected attribute:
\begin{equation}
z_s \;=\; \psi(s).
\label{eq:zs_sketch}
\end{equation}
For each feature coordinate $j$, compute (i) signed probability gap $\Delta P_j$ and (ii) $\mathrm{HSIC}_j$, then rank by
\begin{equation}
S_j \;=\; f\!\big(|\Delta P_j|,\mathrm{HSIC}_j\big),
\label{eq:scoreSj}
\end{equation}
and derive per-party sets $(M^{(p)},P^{(p)},N^{(p)})$ via fractions $(\rho_M,\rho_P)$:
\begin{equation}
P^{(p)} \;=\; \Bigl\{ j \in M^{(p)} : S_j \text{ lies in the top } \rho_P \text{ of } \{S_k : k \in M^{(p)}\} \Bigr\}.
\label{eq:proxy_set}
\end{equation}
Fractions are selected on validation, then fixed across seeds/methods for that dataset.

\noindent\textbf{Governance note.} The partition into permissible mediators \(M\) and impermissible proxies \(P\subset M\) is a normative policy choice; we make this choice auditable via a concrete protocol in Appendix~\ref{app:policy-audit}.

\subsection{Differential privacy and secure aggregation}
\label{app:dp-secagg}

\paragraph{Scope and threat model.}
We use differential privacy \emph{only} to privatize the released sketch used for mediator/proxy screening in mask discovery. This provides record-level protection for the contingency tables (and any downstream statistics computed from them) against an observer who sees the released sketch. It does \emph{not} provide end-to-end DP for learned representations, model parameters, or gradients, and it does not rule out inference attacks outside the evaluated threat model.

\paragraph{DP sketch for mediator discovery.}
A trusted holder of the sensitive attribute (or a trusted enclave) computes and releases the DP sketch; the coordinating server and other parties never observe raw $s$. For feature $j$, discretise into $K$ bins and build contingency table $c_j\in\mathbb{R}^{2\times K}$, clip:
\begin{equation}
\bar{c}_j \;=\; c_j \cdot \min\!\left(1, \frac{S}{\|c_j\|_2}\right),
\label{eq:clip_counts}
\end{equation}
then add Gaussian noise:
\begin{equation}
\tilde{c}_j \;=\; \bar{c}_j + \mathcal{N}\!\left(0, \sigma_{\mathrm{sketch}}^2 S^2 I\right),
\label{eq:gauss_mech}
\end{equation}
and compute $\Delta P_j,\mathrm{HSIC}_j$ from $\tilde{c}_j$.

\begin{proof}[Proof of Proposition~\ref{prop:dp_sketch}]
Under addition or removal of one record, the contingency table $c_j\in\mathbb{R}^{2\times K}$ changes in exactly one cell by $\pm 1$, so the $\ell_2$-sensitivity of the unclipped table is $1$. After clipping (Eq.~\eqref{eq:clip_counts}), the $\ell_2$-norm of $\bar{c}_j$ is bounded by $S$, and the $\ell_2$-sensitivity of the clipped table under neighboring datasets is at most $S$ (since each record contributes at most $S$ to the clipped output). The Gaussian mechanism (Eq.~\eqref{eq:gauss_mech}) adds noise with standard deviation $\sigma_{\mathrm{sketch}} \cdot S$ per coordinate. By the standard Gaussian mechanism guarantee (Theorem A.1 of Dwork and Roth~\cite{dwork2014algorithmic}), releasing $\tilde{c}_j$ satisfies $(\varepsilon,\delta)$-DP with $\varepsilon = (1/\sigma_{\mathrm{sketch}})\sqrt{2\ln(1.25/\delta)}$ for any $\delta \in (0,1)$.
\end{proof}

\noindent When masks are refreshed, the same mechanism is re-run; we cap the number of refreshes, and each refresh consumes an additional $(\varepsilon,\delta)$ budget. A full composition accountant (e.g., via R\'{e}nyi DP or the moments accountant) is outside the scope of this work but is a natural extension.

\paragraph{Secure aggregation.}
Secure aggregation protects per-client updates from being inspected individually by the coordinating server. Client gradient $g_t^{(c)}$ is masked as
\begin{equation}
u_t^{(c)} \;=\; g_t^{(c)} + r_t^{(c)},
\label{eq:secagg_send}
\end{equation}
with masks satisfying $\sum_{c=1}^m r_t^{(c)}=0$, so the server recovers only
\begin{equation}
\sum_{c=1}^m u_t^{(c)} \;=\; \sum_{c=1}^m g_t^{(c)}.
\label{eq:secagg_sum}
\end{equation}
We add no DP noise to training gradients or model updates in the main experiments; DP is applied only to the mask-discovery sketch above.

\subsection{SCC--VFL model}
\label{app:sccvfl-model}

\paragraph{Masked conditional generators.}
Client generator edits only mediators:
\begin{equation}
x^{\mathrm{cf},(c)}_{M} \;=\; g_{\varphi}^{(c)}\!\big(x^{(c)}_{N}, x^{(c)}_{M}, c_i, s'\big),
\qquad
x^{\mathrm{cf},(c)}_{N} \;=\; x^{(c)}_{N}.
\label{eq:gen_edit}
\end{equation}
We instantiate $g_\varphi^{(c)}$ as a cVAE:
\begin{equation}
z^{(c)} \sim q_{\varphi}^{(c)}\!\big(z^{(c)} \mid x^{(c)}_{N}, x^{(c)}_{M}, s\big),
\qquad
x^{\mathrm{cf},(c)}_{M} = d_{\varphi}^{(c)}\!\big(z^{(c)}, x^{(c)}_{N}, c_i, s'\big).
\label{eq:cvae}
\end{equation}

\paragraph{Proxy adversaries and gradient reversal.}
Per-party adversary with GRL:
\begin{equation}
\tilde{u}_i^{(p)} = \mathrm{GRL}_{\lambda_{\mathrm{grl}}}\big(u_i^{(p)}\big),
\qquad
\hat{s}_i^{(p)} = a_{\omega}^{(p)}\!\big(\tilde{u}_i^{(p)}\big),
\label{eq:grl_adv}
\end{equation}
and total adversarial loss
\begin{equation}
L_{\mathrm{adv}} \;=\; \sum_{p=1}^{m} L_{\mathrm{adv}}^{(p)}.
\label{eq:adv_sum}
\end{equation}

\paragraph{Objective.}
\begin{equation}
L \;=\;
L_{\mathrm{cls}}
+ \lambda_{\mathrm{cons}} L_{\mathrm{cons}}
+ \lambda_{\mathrm{gen}} L_{\mathrm{gen}}
+ \lambda_{\mathrm{adv}} L_{\mathrm{adv}}.
\label{eq:full_obj}
\end{equation}

\subsection{Baseline implementations}
\label{app:baseline-impl}

Baselines are ablations: same backbone Eqs.~\eqref{eq:mlp_backbone}--\eqref{eq:server_head}, same partitions (Table~\ref{tab:client-partitions}), same hyperparameters (Table~\ref{tab:hparams-common}), and activation pattern in Table~\ref{tab:hparams-method}.

\textbf{Adv-NoMask.} Task loss plus adversarial debiasing (no mask discovery; adversary active).

\textbf{Uniform-CF.} Generators edit all coordinates; uses $L_{\mathrm{cls}}+\lambda_{\mathrm{cons}}L_{\mathrm{cons}}$; no adversary.

\textbf{Policy-blind mask.} Uses mediator mask $M$ but no proxy handling; generators and consistency active; adversary off.

\textbf{Server only consistency.} No generators; server perturbs/shuffles mediator coordinates and applies consistency; adversary off.

\subsection{Hyperparameter settings}
\label{app:hparams}

Shared settings are in Table~\ref{tab:hparams-common}. Active components per method are in Table~\ref{tab:hparams-method}.

\begin{table}[t]
  \caption{Shared architectural and optimisation hyperparameters per dataset. These values are held fixed across SCC--VFL and all baselines.}
  \label{tab:hparams-common}
  \centering
  \small
  \setlength{\tabcolsep}{4pt}
  \renewcommand{\arraystretch}{1.15}
  \begin{tabular}{lccc}
    \toprule
    Setting & German Credit & Heart Disease & COMPAS Cox \\
    \midrule
    \# feature clients $m$ & 3 & 3 & 3 \\
    Encoder hidden dim $d_z$ & 64 & 32 & 32 \\
    Encoder layers & 2 & 2 & 2 \\
    Dropout (encoder) & 0.05 & 0.10 & 0.10 \\
    Classifier output units & 2 & 2 & 2 \\
    Generator type & cVAE & cVAE & cVAE \\
    Generator latent dim & 8 & 16 & 8 \\
    Optimiser (backbone) & AdamW & Adam & Adam \\
    Learning rate (backbone) & 0.015 & 0.010 & 0.005 \\
    Learning rate (generators) & 0.010 & 0.010 & 0.005 \\
    Learning rate (adversary) & 0.005 & 0.005 & 0.003 \\
    Batch size & full batch & full batch & full batch \\
    Max epochs (baselines) & 300 & 200 & 80 \\
    Max epochs (SCC--VFL) & 300 & 200 & 150 \\
    Warmup epochs (no $L_{\mathrm{cons}}$) & 40 & 30 & 25 \\
    Early stopping patience & 35 & 25 & 25 \\
    Number of seeds & 30 & 30 & 30 \\
    \bottomrule
  \end{tabular}
\end{table}

\begin{table}[t]
  \caption{Loss terms and generator masks per method. All methods share the same backbone and optimiser hyperparameters; only the active components differ.}
  \label{tab:hparams-method}
  \centering
  \small
  \setlength{\tabcolsep}{4pt}
  \renewcommand{\arraystretch}{1.15}
  \begin{tabular}{lcccccc}
    \toprule
    Method & Gen. mask & $L_{\mathrm{cls}}$ & $L_{\mathrm{cons}}$ & $L_{\mathrm{gen}}$ & $L_{\mathrm{adv}}$ & Adv. on $h$ \\
    \midrule
    Adv-NoMask   & none               & \checkmark & $\times$     & $\times$     & \checkmark & \checkmark \\
    Uniform CF & all coords & \checkmark & \checkmark & \checkmark & $\times$ & $\times$ \\
    Policy blind mask & $M$ & \checkmark & \checkmark & \checkmark & $\times$ & $\times$ \\
    Server only consistency & random $M$ shuffle & \checkmark & \checkmark & $\times$ & $\times$ & $\times$ \\
    SCC--VFL (ours) & $M$ per client & \checkmark & \checkmark & \checkmark & \checkmark & \checkmark \\
    \bottomrule
  \end{tabular}
\end{table}

\section{Policy Specification and Audit Protocol for \(N/M/P\)}
\label{app:policy-audit}

SCC-VFL distinguishes \emph{editable mediators} \(M\) from \emph{impermissible proxies} \(P \subset M\). This split is a normative policy decision: \(M\) encodes which descendants of \(s\) are considered acceptable pathways for counterfactual edits, while \(P\) marks descendants judged unacceptable to rely on (even if predictive). We therefore make the \(N/M/P\) specification explicit, versioned, and contestable.


\paragraph{Domain examples (illustrative).}
\textbf{German Credit (protected: age).}
A lending policy board may treat financially grounded, action-relevant factors as permissible mediators (e.g., repayment history signals) while flagging demographic-adjacent or structurally discriminatory signals as proxies if they encode age-related opportunity differences. The Policy Card records which feature groups are placed in \(M\) vs. \(P\), and why, and documents the appeal path for disputed assignments.

\textbf{Heart Disease (protected: sex).}
A clinical governance group may allow physiological variables that plausibly sit on causal pathways for disease risk as mediators, while treating administrative or measurement artifacts that encode sex in a non-clinical way as proxies. Reviews are triggered when clinical guidelines change or when the patient population shifts.

\textbf{COMPAS Cox (protected: race).}
A justice-policy panel may classify variables tied to enforcement intensity or structural inequities as proxies (even if predictive), while allowing narrowly defined, policy-justified pathways as mediators if explicitly permitted by the deployment context. The appeal process supports challenges to specific role assignments and logs any updates under a new policy version.

\begin{tcolorbox}[title={Policy Specification and Audit Protocol (auditable \(N/M/P\))},fonttitle=\bfseries]
\small
\textbf{A1. Policy owner (who decides).}
A named governance group defines and approves the mediator/proxy policy: regulator or compliance lead (if applicable, consistent with non-discrimination law~\cite{weerts2023algorithmic}), an institutional policy board, domain experts, and at least one affected-stakeholder representative.

\textbf{A2. Inputs (what is being decided).}
A dataset-specific \emph{feature dictionary} (definitions, measurement process, manipulability, known confounders, and known correlations with \(s\)) and the SCC-VFL discovery outputs (DP-ranked coordinates).

\textbf{A3. Evidence (what counts).}
Accepted evidence includes written statutes or institutional policy, domain guidelines, documented stakeholder input, and empirical checks (feature meaning, stability under shift, and whether a feature is a plausible descendant of \(s\) in the application context).

\textbf{A4. Decision rule (how \(N/M/P\) is produced).}
Start from the DP-ranked candidates as \emph{candidate descendants}. The policy owner assigns:
(i) \(N\): coordinates treated as non-descendants (held fixed),
(ii) \(M\): permissible descendants (editable under CF),
(iii) \(P \subset M\): impermissible descendants (treated as leakage-risk and not relied on).
Any override to the ranking must include a written rationale.

\textbf{A5. Artifacts (what is logged).}
Create a Policy Card per dataset, modeled after Model Cards~\cite{mitchell2019model} and Datasheets for Datasets~\cite{gebru2021datasheets}, with: version ID, date, owners, rationale per role decision, evidence links, and a machine-readable $N/M/P$ mask file used in training and evaluation.

\textbf{A6. Contestability (appeals).}
Provide a documented process for an individual or auditor to dispute a role assignment. The appeal triggers: (i) review of the feature dictionary entry, (ii) re-evaluation of evidence, and (iii) a logged decision with an updated version if changed.

\textbf{A7. Review cadence (when it is revisited).}
Re-approve the policy on a fixed schedule (e.g., quarterly) and additionally upon major distribution shift, feature set changes, or policy changes. Each review re-runs discovery on the new data slice and records deltas in \(N/M/P\).
\end{tcolorbox}

\section{Privacy and Robustness Evaluation}
\label{app:privacy_robustness}

This appendix evaluates (i) \emph{representation leakage} via an attribute inference attacker trained on latent vectors, and (ii) \emph{input robustness} via a constrained adversary that perturbs only a designated feature subspace.
Threat model and notation are in Section~\ref{app:privacy_robustness}; AIA in Section~\ref{app:aia_theory}; subspace-PGD in Section~\ref{app:pgd_theory}; results are reported in Tables~\ref{tab:aia_sr} and~\ref{tab:pgd_sr}.

\subsection{Threat Model and Notation}

We consider a trained model with encoder $h_\phi$ and prediction head $f_\theta$ under the VFL threat surface that spans feature inference~\cite{luo2021featureinference, jin2021cafe} and label inference~\cite{fu2022label} attacks. For entity $i$ with input $x_i$, the encoder produces a latent representation $h_i\in\mathbb{R}^{d_h}$ and the head outputs a predictive distribution over labels $\mathcal{Y}$:
\begin{equation}
h_i = h_\phi(x_i) \in \mathbb{R}^{d_h},
\quad
\hat{y}_i = f_\theta(h_i) \in \Delta(\mathcal{Y}).
\label{eq:threat_notation}
\end{equation}
The protected attribute is $s_i\in\{0,1\}$. We use $\hat{y}_i$ (overloading notation) for the predicted class obtained by the argmax over class probabilities:
\begin{equation}
\hat{y}_i = \arg\max_{y} f_\theta(h_i)_y.
\label{eq:pred_class}
\end{equation}
The adversary is post-hoc (after training) and is evaluated by success rate (SR, in \%) where higher SR indicates stronger attack and thus weaker privacy/robustness.

\subsection{Attribute Inference Attack (AIA)}
\label{app:aia_theory}

AIA measures how much information about $s$ remains in the learned representation $h_i$. The attacker trains a binary classifier
\begin{equation}
g_\psi : \mathbb{R}^{d_h} \to [0,1],
\label{eq:aia_model}
\end{equation}
where $g_\psi(h_i)$ estimates $\Pr(s_i=1\mid h_i)$. Given attacker training data $\{(h_i,s_i)\}_{i=1}^n$, parameters $\psi$ are fit by minimizing binary cross-entropy:
\begin{equation}
\mathcal{L}_{\text{AIA}}(\psi)
= \frac{1}{n} \sum_{i=1}^{n}
\ell_{\text{BCE}}\big(g_\psi(h_i), s_i\big),
\quad
\ell_{\text{BCE}}(p,s)
= - s \log p - (1-s) \log (1-p).
\label{eq:aia_loss}
\end{equation}
The attacker predicts $\hat{s}_i$ by thresholding at $1/2$:
\begin{equation}
\hat{s}_i = \mathbb{I}\{ g_\psi(h_i) \ge 1/2\}.
\label{eq:aia_rule}
\end{equation}
For an attacker training budget $T$ (epochs), we report the attack success rate on a balanced test split $\mathcal{I}_{\text{test}}$:
\begin{equation}
\text{AIA-SR}(T)
= \frac{1}{|\mathcal{I}_{\text{test}}|}
\sum_{i \in \mathcal{I}_{\text{test}}}
\mathbb{I}\{\hat{s}_i^{(T)} = s_i\} \times 100\,\%.
\label{eq:aia_sr}
\end{equation}
As a reference, $\text{AIA-SR}^\star$ denotes the Bayes-optimal success rate achievable by the optimal classifier $g^\star$ under the induced representation distribution:
\begin{equation}
\text{AIA-SR}^\star
= \mathbb{E}\big[\mathbb{I}\{g^\star(h) = s\}\big].
\label{eq:aia_bayes}
\end{equation}
Lower AIA-SR indicates less $s$-signal in $h$.

\subsection{Subspace-Constrained PGD Attack}
\label{app:pgd_theory}

Subspace-PGD measures robustness to targeted feature manipulation. The attacker chooses a perturbation $\delta_i$ that is (i) supported only on a designated coordinate subset $S$ (the ``attackable'' subspace) and (ii) bounded in $\ell_\infty$ norm by $\epsilon$:
\begin{equation}
\text{supp}(\delta_i) \subseteq S,
\qquad
\|\delta_i\|_\infty \le \epsilon.
\label{eq:pgd_constraint}
\end{equation}
Starting from $\delta_i^{(0)}=0$, PGD iteratively updates $\delta_i$ using signed gradients of the task loss (with respect to the input) and then projects back to the feasible set:
\begin{equation}
\delta_i^{(t+1)}
= \Pi_{\mathcal{C}_\epsilon}\!\left(
\delta_i^{(t)}
+ \alpha \,\text{sign}\!\big(\nabla_{x} \ell_{\text{task}}
(f_\theta(h_\phi(x_i + \delta_i^{(t)})), \hat{y}_i)\big)
\right),
\label{eq:pgd_update}
\end{equation}
where $\alpha$ is the step size and $\Pi_{\mathcal{C}_\epsilon}$ projects onto
\begin{equation}
\mathcal{C}_\epsilon
= \big\{\delta \in \mathbb{R}^{d_x} :
\text{supp}(\delta) \subseteq S,\,
\|\delta\|_\infty \le \epsilon
\big\}.
\label{eq:pgd_set}
\end{equation}
After $T_{\text{PGD}}$ steps, the adversarial example is
\begin{equation}
x_i^{\text{adv}} = x_i + \delta_i^{(T_{\text{PGD}})}.
\label{eq:pgd_advex}
\end{equation}
The PGD success rate reports the fraction of inputs whose predicted class changes under attack:
\begin{equation}
\text{PGD-SR}(\epsilon)
= \frac{1}{n} \sum_{i=1}^{n}
\mathbb{I}\{\hat{y}_i^{\text{adv}} \ne \hat{y}_i\} \times 100\,\%.
\label{eq:pgd_sr}
\end{equation}
Lower PGD-SR indicates stronger robustness to targeted subspace perturbations.

\subsection{Implementation Details}
\label{app:attack_impl}

Table~\ref{tab:attack_impl} summarises the attack implementations and evaluation settings used for AIA and Subspace-PGD.

\begin{table}[t]
  \centering
  \caption{Attack implementation details and evaluation protocol.}
  \label{tab:attack_impl}
  \small
  \setlength{\tabcolsep}{5pt}
  \renewcommand{\arraystretch}{1.15}
  \begin{tabular}{l p{10.2cm}}
    \toprule
    Component & Setting \\
    \midrule
    AIA data & Balanced split over $s$; compute and freeze representations $h_\phi(x)$ \\
    AIA attacker & Train $g_\psi$ for $T \in \{10,20,40,80\}$ epochs \\
    AIA metric & Report AIA-SR via Eq.~\eqref{eq:aia_sr}, aggregated over $30$ seeds \\
    \midrule
    PGD subspace & Set sensitive subspace $S=M$ \\
    PGD steps & $T_{\text{PGD}}=20$ with step size $\alpha=\epsilon/5$ \\
    PGD radii & $\epsilon \in \{0.02,0.05,0.10,0.20\}$ \\
    PGD metric & Report PGD-SR via Eq.~\eqref{eq:pgd_sr}, aggregated over $30$ seeds \\
    \bottomrule
  \end{tabular}
\end{table}

\subsection{Empirical Results}
\label{app:attack_results}

AIA-SR is reported in Table~\ref{tab:aia_sr}. PGD-SR is reported in Table~\ref{tab:pgd_sr}.

\begin{table}[t]
    \centering
    \caption{Attribute inference attack success rate (AIA-SR, \%) as a function of attacker training budget (epochs), averaged over 30 seeds. Lower is better.}
    \label{tab:aia_sr}
    \begin{tabular}{lcccc}
        \hline
        \textbf{Method}          & \textbf{10 epochs}     & \textbf{20 epochs}     & \textbf{40 epochs}     & \textbf{80 epochs}     \\
        \hline
        Adv-NoMask               & $54.03 \pm 4.57$       & $56.15 \pm 5.35$       & $59.06 \pm 6.00$       & $62.52 \pm 5.55$       \\
        Uniform-CF               & $52.64 \pm 4.31$       & $55.41 \pm 4.61$       & $57.93 \pm 4.63$       & $62.51 \pm 4.26$       \\
        Policy-blind Mask        & $53.11 \pm 5.03$       & $55.47 \pm 5.55$       & $58.35 \pm 6.50$       & $62.04 \pm 5.71$       \\
        Server-Consistency       & $52.84 \pm 4.03$       & $53.49 \pm 3.98$       & $59.79 \pm 5.42$       & $62.27 \pm 4.91$       \\
        SCC-VFL (ours)           & $53.08 \pm 4.38$       & $53.45 \pm 4.67$       & $56.17 \pm 6.39$       & $57.89 \pm 5.72$       \\
        \hline
    \end{tabular}
\end{table}

\begin{table}[t]
    \centering
    \caption{Subspace-constrained PGD attack success rate (PGD-SR, \%) as a function of perturbation radius $\epsilon$, averaged over 30 seeds. Lower is better.}
    \label{tab:pgd_sr}
    \begin{tabular}{lcccc}
        \hline
        \textbf{Method}          & \boldmath$\epsilon=0.02$ & \boldmath$\epsilon=0.05$ & \boldmath$\epsilon=0.10$ & \boldmath$\epsilon=0.20$ \\
        \hline
        Adv-NoMask               & $2.24 \pm 0.90$         & $5.97 \pm 1.43$         & $11.41 \pm 1.80$        & $21.67 \pm 2.85$        \\
        Uniform-CF               & $1.94 \pm 0.81$         & $5.14 \pm 1.08$         & $10.36 \pm 1.64$        & $20.64 \pm 3.05$        \\
        Policy-blind Mask        & $2.09 \pm 0.67$         & $5.36 \pm 1.41$         & $10.37 \pm 1.96$        & $20.11 \pm 2.75$        \\
        Server-Consistency       & $1.28 \pm 0.57$         & $3.50 \pm 1.15$         & $7.54 \pm 1.85$         & $15.02 \pm 2.89$        \\
        SCC-VFL (ours)           & $1.03 \pm 0.61$         & $2.44 \pm 0.93$         & $4.53 \pm 1.40$         & $9.30 \pm 2.29$         \\
        \hline
    \end{tabular}
\end{table}


\section{Runtime and Communication Overhead}
\label{app:runtime-comm}

We report wall-clock runtime and communication for SCC-VFL on German Credit (IID), aggregated as mean$\pm$std over 30 random seeds (CPU runs). Over training, SCC-VFL requires $1.0251\pm0.2233$ seconds total, corresponding to $0.0218\pm0.0032$ seconds per epoch and $7.1015\pm1.0378$ ms per optimization step, with $40.1333\pm4.7380$ epochs executed. Communication totals (sum over training) are $2.87$ MB for model-side messages and $124.69$--$146.51$ MB for feature-side messages across the shown seeds (batch size $256$); the per-seed table indicates a stable model communication footprint and feature communication that scales with the number of steps/epochs. These measurements quantify the additional passes and per-party generator components in SCC-VFL in a cross-silo style setting: while model communication remains small, feature-side traffic dominates and should be considered when bandwidth is constrained.

\section{Sensitivity Analysis}
\label{app:sensitivity}

German Credit sensitivity is studied under the same protocol as the main results (fixed split sizes, identical pipeline, 30 seeds). Each sweep changes one hyperparameter; report mean$\pm$std of Accuracy, LogLoss, SCG, and FR(\%). Tables~\ref{tab:sens_topfrac}--\ref{tab:sens_lamgen} contain all results.

\subsection{Mediator Mask Threshold Sensitivity ($\tau_M$)}
\label{app:sens_topfrac}

Mediator set (top-$\tau_M$ by DP ranking score $S_j$):
\begin{equation}
M \;=\; \{ j \in [d] \;:\; S_j \text{ is in the top } \tau_M \text{ fraction of } \{S_\ell\}_{\ell=1}^d \},
\end{equation}
with $S_j$ defined from $|\Delta P_j|$ and $\mathrm{HSIC}_j$ (Section~\ref{subsec:mask}). Table~\ref{tab:sens_topfrac} shows moderate utility variation and consistently low SCG/FR; $\tau_M=0.33$ minimizes SCG and yields FR$=0$ (within resolution).

\begin{table}[t]
\caption{Sensitivity to mediator mask threshold $\tau_M$ (top\_frac). Other hyperparameters fixed.}
\label{tab:sens_topfrac}
\centering
\small
\begin{tabular}{lcccc}
\toprule
Setting & Acc.$\uparrow$ & LogLoss$\downarrow$ & SCG$\downarrow$ & FR(\%)$\downarrow$ \\
\midrule
$\tau_M=0.15$ & $0.7353\pm0.0138$ & $0.5738\pm0.0117$ & $0.0051\pm0.0041$ & $0.13\pm0.16$ \\
$\tau_M=0.25$ & $0.7020\pm0.0045$ & $0.5792\pm0.0195$ & $0.0068\pm0.0074$ & $0.13\pm0.16$ \\
$\tau_M=0.33$ & $0.7187\pm0.0233$ & $0.5780\pm0.0284$ & $\mathbf{0.0027\pm0.0011}$ & $\mathbf{0.00\pm0.00}$ \\
$\tau_M=0.45$ & $0.7180\pm0.0088$ & $0.5742\pm0.0148$ & $0.0083\pm0.0067$ & $0.33\pm0.30$ \\
$\tau_M=0.60$ & $0.7140\pm0.0083$ & $\mathbf{0.5714\pm0.0130}$ & $0.0053\pm0.0056$ & $0.27\pm0.39$ \\
\bottomrule
\end{tabular}
\
\end{table}

\subsection{Counterfactual Edit Magnitude Sensitivity ($\gamma_{\mathrm{cf}}$)}
\label{app:sens_cfscale}

Edit scale for mediator-only perturbations:
\begin{equation}
\tilde{x} \;=\; x \;+\; \gamma_{\mathrm{cf}} \cdot \Pi_{M}\big(\Delta(x)\big),
\end{equation}
where $\Pi_M$ projects to mediator coordinates and $\Delta(x)$ is generator output. Table~\ref{tab:sens_cfscale} shows utility invariant across the tested range; SCG increases smoothly with $\gamma_{\mathrm{cf}}$; FR remains near zero until the largest scales.

\begin{table}[t]
\caption{Sensitivity to counterfactual edit magnitude $\gamma_{\mathrm{cf}}$ (cf\_scale). Other hyperparameters fixed.}
\label{tab:sens_cfscale}
\centering
\small
\begin{tabular}{lcccc}
\toprule
Setting & Acc.$\uparrow$ & LogLoss$\downarrow$ & SCG$\downarrow$ & FR(\%)$\downarrow$ \\
\midrule
$\gamma_{\mathrm{cf}}=0.05$ & $0.7187\pm0.0233$ & $0.5780\pm0.0284$ & $\mathbf{0.0015\pm0.0008}$ & $\mathbf{0.00\pm0.00}$ \\
$\gamma_{\mathrm{cf}}=0.10$ & $0.7187\pm0.0233$ & $0.5780\pm0.0284$ & $0.0019\pm0.0010$ & $\mathbf{0.00\pm0.00}$ \\
$\gamma_{\mathrm{cf}}=0.20$ & $0.7187\pm0.0233$ & $0.5780\pm0.0284$ & $0.0027\pm0.0011$ & $\mathbf{0.00\pm0.00}$ \\
$\gamma_{\mathrm{cf}}=0.35$ & $0.7187\pm0.0233$ & $0.5780\pm0.0284$ & $0.0039\pm0.0011$ & $0.07\pm0.13$ \\
$\gamma_{\mathrm{cf}}=0.50$ & $0.7187\pm0.0233$ & $0.5780\pm0.0284$ & $0.0050\pm0.0011$ & $0.07\pm0.13$ \\
\bottomrule
\end{tabular}
\end{table}

\subsection{Consistency Weight Sensitivity ($\lambda_{\mathrm{cons}}$)}
\label{app:sens_lamcons}

Objective with scheduled max weight $\lambda_{\mathrm{cons}}$:
\begin{equation}
\mathcal{L} \;=\; \mathcal{L}_{\mathrm{task}} \;+\; \lambda_{\mathrm{cons}} \cdot \mathcal{L}_{\mathrm{cons}} \;+\; \lambda_{\mathrm{adv}} \cdot \mathcal{L}_{\mathrm{adv}} \;+\; \lambda_{\mathrm{gen}} \cdot \mathcal{L}_{\mathrm{gen}},
\end{equation}
where $\mathcal{L}_{\mathrm{cons}}$ enforces agreement under masked edits. Table~\ref{tab:sens_lamcons} shows utility stable; SCG rises mildly with $\lambda_{\mathrm{cons}}$; FR stays at 0.

\begin{table}[t]
\caption{Sensitivity to consistency weight $\lambda_{\mathrm{cons}}$ (lam\_cons\_max). Other hyperparameters fixed.}
\label{tab:sens_lamcons}
\small
\begin{tabular}{lcccc}
\toprule
Setting & Acc.$\uparrow$ & LogLoss$\downarrow$ & SCG$\downarrow$ & FR(\%)$\downarrow$ \\
\midrule
$\lambda_{\mathrm{cons}}=0.0$ & $0.7187\pm0.0233$ & $0.5781\pm0.0285$ & $\mathbf{0.0017\pm0.0002}$ & $\mathbf{0.00\pm0.00}$ \\
$\lambda_{\mathrm{cons}}=0.4$ & $0.7187\pm0.0233$ & $0.5780\pm0.0284$ & $0.0020\pm0.0003$ & $\mathbf{0.00\pm0.00}$ \\
$\lambda_{\mathrm{cons}}=0.8$ & $0.7187\pm0.0233$ & $0.5780\pm0.0284$ & $0.0024\pm0.0007$ & $\mathbf{0.00\pm0.00}$ \\
$\lambda_{\mathrm{cons}}=1.2$ & $0.7187\pm0.0233$ & $0.5780\pm0.0284$ & $0.0027\pm0.0011$ & $\mathbf{0.00\pm0.00}$ \\
$\lambda_{\mathrm{cons}}=1.6$ & $0.7187\pm0.0233$ & $0.5780\pm0.0284$ & $0.0031\pm0.0014$ & $\mathbf{0.00\pm0.00}$ \\
\bottomrule
\end{tabular}
\end{table}

\subsection{Adversary Weight Sensitivity ($\lambda_{\mathrm{adv}}$)}
\label{app:sens_lamadv}

We vary $\lambda_{\mathrm{adv}}$ holding others fixed. Table~\ref{tab:sens_lamadv} shows low-to-moderate $\lambda_{\mathrm{adv}}$ maintains FR$=0$ and good utility; larger $\lambda_{\mathrm{adv}}$ slightly increases FR and can reduce utility.

\begin{table}[t]
\caption{Sensitivity to adversary weight $\lambda_{\mathrm{adv}}$. Other hyperparameters fixed.}
\label{tab:sens_lamadv}
\centering
\small
\begin{tabular}{lcccc}
\toprule
Setting & Acc.$\uparrow$ & LogLoss$\downarrow$ & SCG$\downarrow$ & FR(\%)$\downarrow$ \\
\midrule
$\lambda_{\mathrm{adv}}=0.00$ & $0.7120\pm0.0150$ & $0.5777\pm0.0185$ & $\mathbf{0.0026\pm0.0012}$ & $\mathbf{0.00\pm0.00}$ \\
$\lambda_{\mathrm{adv}}=0.01$ & $\mathbf{0.7220\pm0.0173}$ & $\mathbf{0.5771\pm0.0176}$ & $0.0027\pm0.0013$ & $\mathbf{0.00\pm0.00}$ \\
$\lambda_{\mathrm{adv}}=0.03$ & $0.7187\pm0.0233$ & $0.5780\pm0.0284$ & $0.0027\pm0.0011$ & $\mathbf{0.00\pm0.00}$ \\
$\lambda_{\mathrm{adv}}=0.06$ & $0.7120\pm0.0133$ & $0.5741\pm0.0225$ & $0.0028\pm0.0013$ & $0.07\pm0.13$ \\
$\lambda_{\mathrm{adv}}=0.10$ & $0.7080\pm0.0096$ & $0.5869\pm0.0115$ & $0.0031\pm0.0020$ & $0.13\pm0.16$ \\
\bottomrule
\end{tabular}
\end{table}

\subsection{Generator Regularization Sensitivity ($\lambda_{\mathrm{gen}}$)}
\label{app:sens_lamgen}

We vary $\lambda_{\mathrm{gen}}$ holding others fixed. Table~\ref{tab:sens_lamgen} shows $\lambda_{\mathrm{gen}}=0$ yields large SCG and non-zero FR; small regularization restores FR$=0$ and reduces SCG without hurting utility.

\begin{table}[t]
\caption{Sensitivity to generator regularization $\lambda_{\mathrm{gen}}$ (lam\_gen). Other hyperparameters fixed.}
\label{tab:sens_lamgen}
\centering
\small
\begin{tabular}{lcccc}
\toprule
Setting & Acc.$\uparrow$ & LogLoss$\downarrow$ & SCG$\downarrow$ & FR(\%)$\downarrow$ \\
\midrule
$\lambda_{\mathrm{gen}}=0.000$ & $0.7207\pm0.0241$ & $0.5786\pm0.0281$ & $0.0294\pm0.0084$ & $0.93\pm0.44$ \\
$\lambda_{\mathrm{gen}}=0.005$ & $0.7187\pm0.0233$ & $0.5780\pm0.0284$ & $0.0037\pm0.0019$ & $\mathbf{0.00\pm0.00}$ \\
$\lambda_{\mathrm{gen}}=0.010$ & $0.7187\pm0.0233$ & $0.5780\pm0.0284$ & $0.0027\pm0.0011$ & $\mathbf{0.00\pm0.00}$ \\
$\lambda_{\mathrm{gen}}=0.020$ & $0.7187\pm0.0233$ & $0.5780\pm0.0284$ & $0.0021\pm0.0006$ & $\mathbf{0.00\pm0.00}$ \\
$\lambda_{\mathrm{gen}}=0.050$ & $0.7187\pm0.0233$ & $0.5780\pm0.0284$ & $\mathbf{0.0018\pm0.0002}$ & $\mathbf{0.00\pm0.00}$ \\
\bottomrule
\end{tabular}

\end{table}

Overall (Tables~\ref{tab:sens_topfrac}--\ref{tab:sens_lamgen}), utility is stable across sweeps; SCG changes smoothly with edit strength/consistency; $\lambda_{\mathrm{gen}}$ is the main safeguard against high SCG and non-zero FR.


\section{Qualitative Analysis}
\label{sec:qualitative}

German Credit IID qualitative views (representative seed 0; patterns consistent with 30-seed aggregates). Figures~\ref{fig:mi_roles}--\ref{fig:scc_training_curves} and Table~\ref{tab:mask_interpret} summarize mask roles, latent-space structure, and training dynamics.

\subsection{Feature roles and DP sketch scores}

We rank features by the DP score $|\Delta P_j|+\mathrm{HSIC}_j$ (Section~\ref{subsec:mask}) and assign roles $N$, $M$, and $P\subset M$. Figure~\ref{fig:mi_roles} shows the highest scores as proxies (red), mid as mediators (green), and low as non-descendants (blue).

\begin{figure}[t]
  \centering
\resizebox{0.8\textwidth}{!}{%
\begin{tikzpicture}
    \begin{axis}[
      xmin=0.3,          
  xmax=19.5,         
  enlarge x limits=false,
  ybar,
      bar width=7pt,
      bar shift=0pt,
      width=0.95\linewidth,
      height=6.0cm,
      ymin=0, ymax=0.36,
      ylabel={DP score $|\Delta P_j| + \mathrm{HSIC}_j$},
      xtick={1,...,19},
      xticklabels={
        \texttt{employment},
        \texttt{existing\_credits},
        \texttt{credit\_history},
        \texttt{telephone},
        \texttt{check\_status},
        \texttt{dependents},
        \texttt{housing},
        \texttt{job},
        \texttt{property},
        \texttt{credit\_amount},
        \texttt{residence\_since},
        \texttt{duration},
        \texttt{purpose},
        \texttt{installment\_rate},
        \texttt{personal\_status\_sex},
        \texttt{other\_debtors},
        \texttt{savings},
        \texttt{foreign\_worker},
        \texttt{other\_install\_plans}
      },
      xticklabel style={rotate=90, anchor=east},
      ymajorgrids=false,
      xmajorgrids=false,
      axis x line*=bottom,
      axis y line*=left,
      legend style={at={(0.98,0.98)}, anchor=north east, draw=none, fill=white, font=\small},
      legend cell align=left
    ]

      \addplot+[fill=blue!20, draw=blue, postaction={pattern=north east lines, pattern color=black}
      ] coordinates {
        (8,0.143) (9,0.120) (10,0.095) (11,0.084) (12,0.069) (13,0.064)
        (14,0.058) (15,0.046) (16,0.041) (17,0.032) (18,0.016) (19,0.008)
      };

      \addplot+[fill=green!50, draw=green!50!black, postaction={pattern=crosshatch, pattern color=black}] coordinates {
        (4,0.214) (5,0.209) (6,0.177) (7,0.151)
      };

      \addplot+[fill=red!50, draw=red, postaction={pattern=dots, pattern color=black}] coordinates {
        (1,0.339) (2,0.334) (3,0.246)
      };

      \legend{Non-descendant, Mediator, Proxy ($P \subset M$)}
    \end{axis}
  \end{tikzpicture}
  }
  \caption{Differentially private mask scores $|\Delta P_j| + \mathrm{HSIC}_j$ for German Credit features, colored by SCC-VFL mask roles: non descendants $N$ (blue), mediators $M$ (green), and proxies $P \subset M$ (red).}
  \label{fig:mi_roles}
\end{figure}
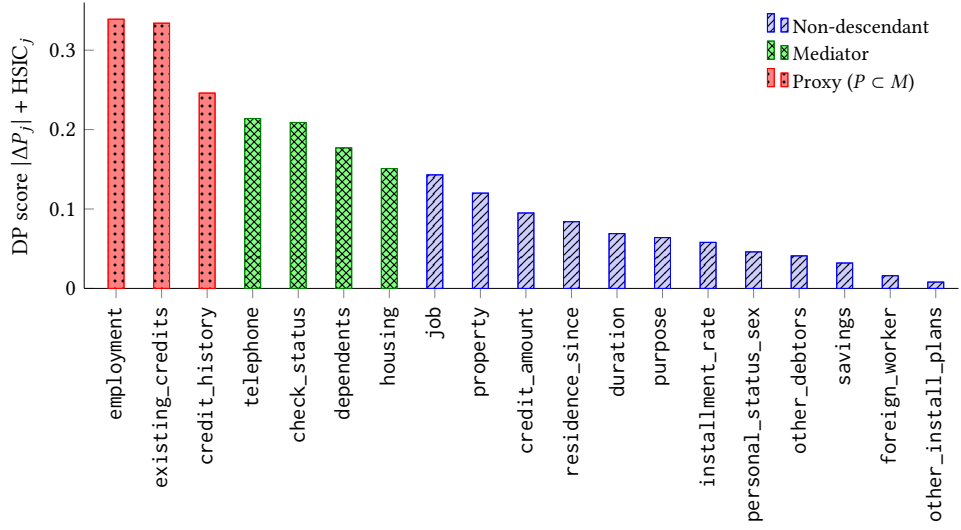

\subsection{Mask interpretability summary}
\label{subsec:mask-interpret}

Table~\ref{tab:mask_interpret} reports a representative ranking of features by $\operatorname{MI}(x_j,s)$ together with their assigned roles. High-$\operatorname{MI}$ coordinates are consistently categorized as proxies ($P$), while mid-range coordinates fall into mediators ($M$). Features with low or near-zero $\operatorname{MI}$ are assigned to non-descendants ($N$). This monotone alignment between dependence strength and role supports the intended semantics of $P/M/N$ and indicates the mask is not arbitrary.

\begin{table}[t]
  \caption{Example mask interpretation on German Credit: high \(\operatorname{MI}\) demographic variables fall into proxies \(P\), mid \(\operatorname{MI}\) structural variables into mediators \(M\), and low \(\operatorname{MI}\) attributes into non-descendants \(N\).}
  \label{tab:mask_interpret}
  \centering
  \small
  \setlength{\tabcolsep}{6pt}
  \renewcommand{\arraystretch}{1.1}
  \begin{tabular}{lcc}
    \toprule
    Feature name & \(\operatorname{MI}(x_j, s)\) & Role \\
    \midrule
    \texttt{personal\_status\_sex} & 0.0539 & Proxy (\(P\)) \\
    \texttt{employment}           & 0.0476 & Proxy (\(P\)) \\
    \texttt{housing}              & 0.0471 & Proxy (\(P\)) \\
    \texttt{other\_install\_plans}& 0.0335 & Mediator (\(M\)) \\
    \texttt{property}             & 0.0267 & Mediator (\(M\)) \\
    \texttt{purpose}              & 0.0226 & Mediator (\(M\)) \\
    \texttt{existing\_credits}    & 0.0215 & Mediator (\(M\)) \\
    \texttt{residence\_since}     & 0.0181 & Non-desc (\(N\)) \\
    \texttt{dependents}           & 0.0170 & Non-desc (\(N\)) \\
    \texttt{credit\_amount}       & 0.0000 & Non-desc (\(N\)) \\
    \bottomrule
  \end{tabular}
\end{table}

\subsection{Latent representation PCA views}

Figure~\ref{fig:pca_methods} compares PCA projections across methods; SCC-VFL shows the strongest mixing between $s$ groups, aligning with reduced sensitive signal in the fused representation.


\begin{figure*}[t]
  \centering

  \begin{subfigure}[t]{0.32\textwidth}
    \centering
    \includegraphics[width=\linewidth]{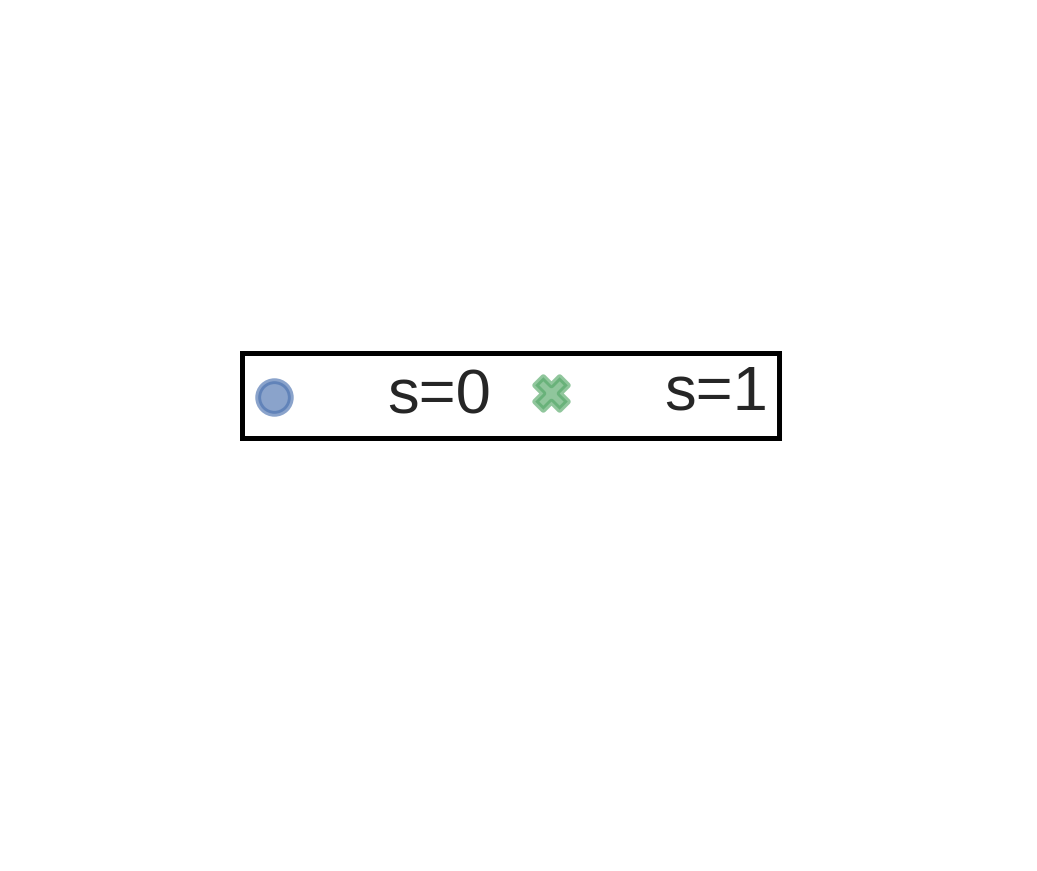}
    \caption*{\footnotesize Legend}
  \end{subfigure}\hfill
  \begin{subfigure}[t]{0.32\textwidth}
    \centering
    \includegraphics[width=\linewidth]{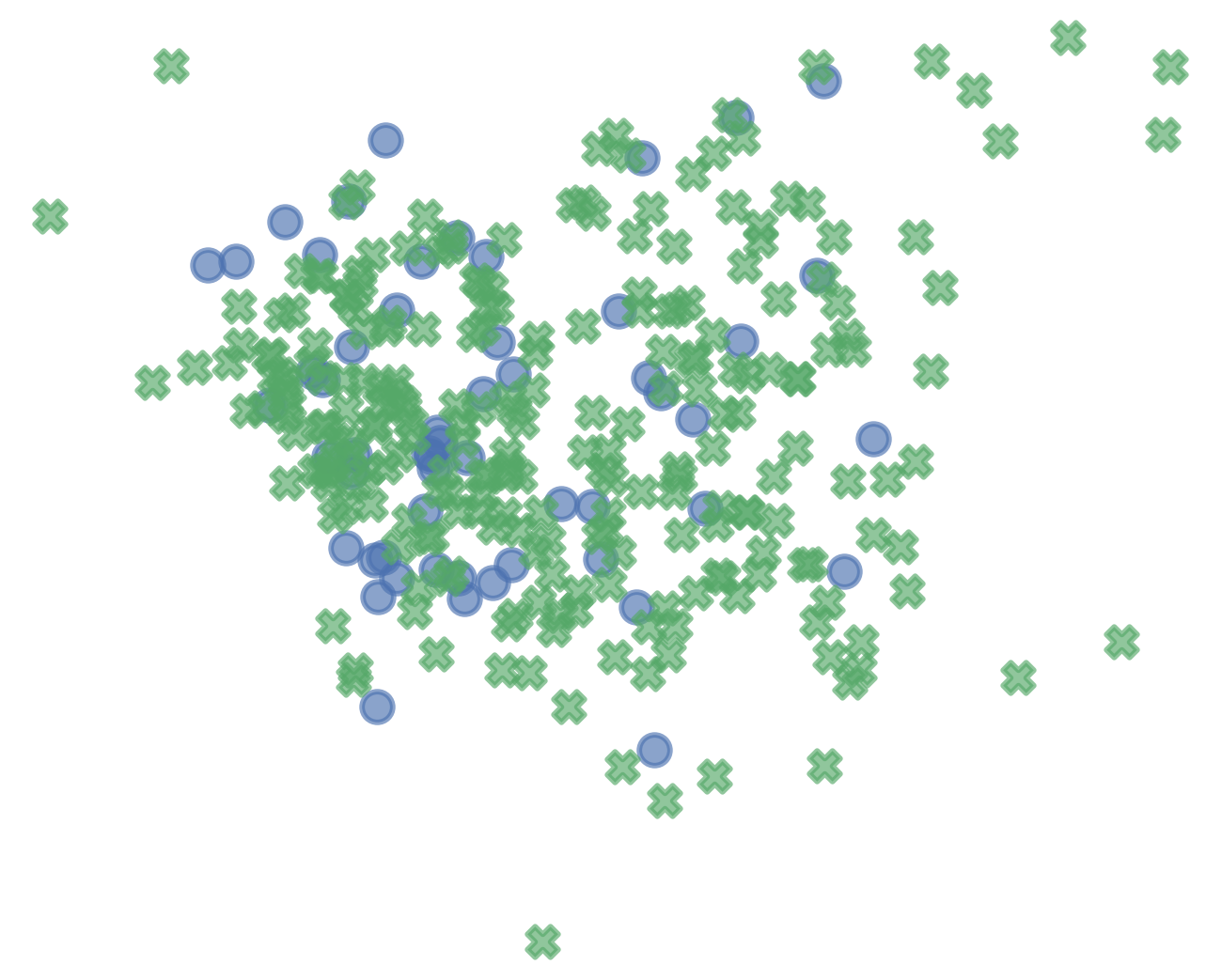}
    \caption{Baseline (no fairness)}
    \label{fig:pca_baseline}
  \end{subfigure}\hfill
  \begin{subfigure}[t]{0.32\textwidth}
    \centering
    \includegraphics[width=\linewidth]{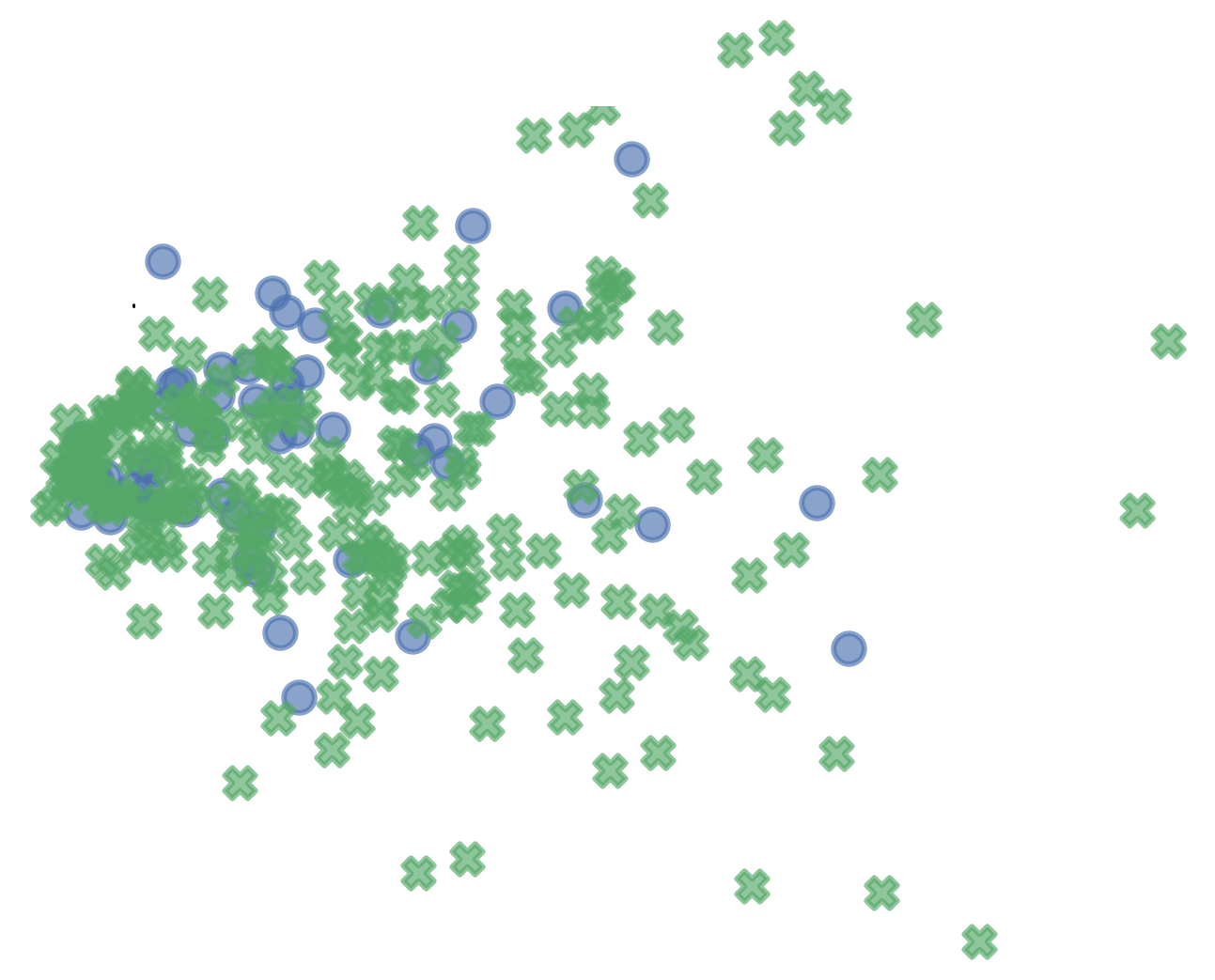}
    \caption{Uniform-CF (na\"ive)}
    \label{fig:pca_uniformcf}
  \end{subfigure}

  \vspace{0.6em}

  \begin{subfigure}[t]{0.32\textwidth}
    \centering
    \includegraphics[width=\linewidth]{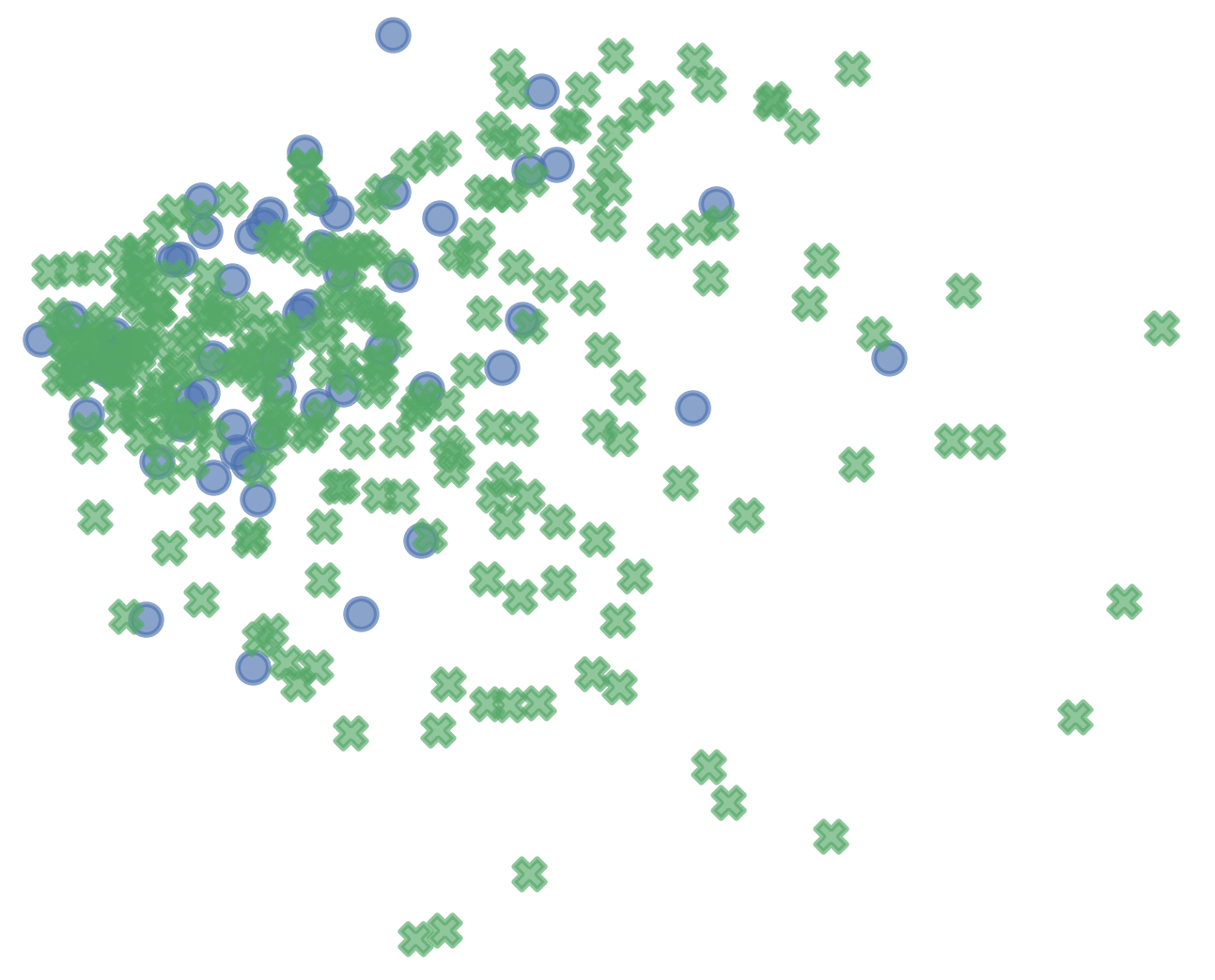}
    \caption{Policy-blind Mask}
    \label{fig:pca_policyblind}
  \end{subfigure}\hfill
  \begin{subfigure}[t]{0.32\textwidth}
    \centering
    \includegraphics[width=\linewidth]{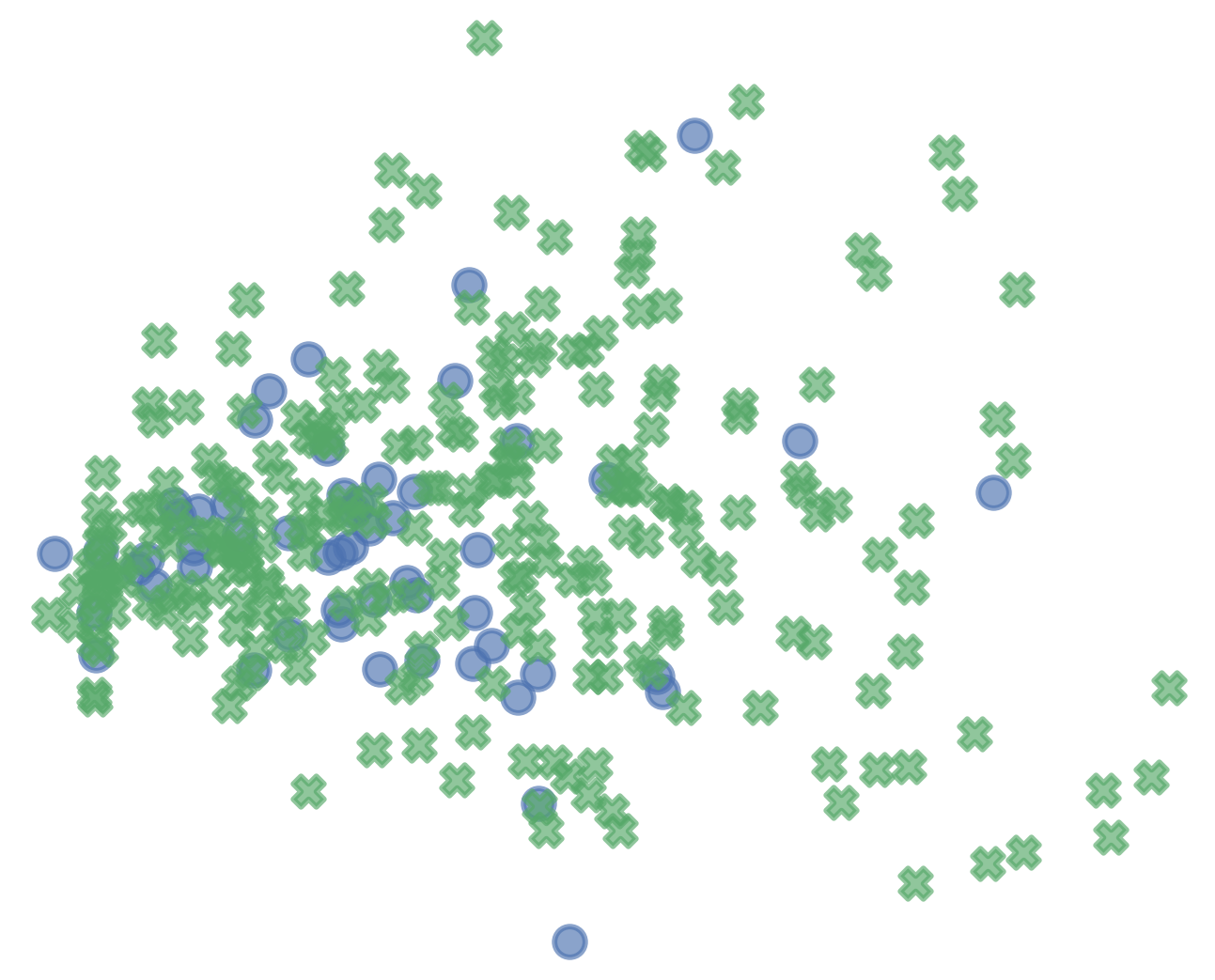}
    \caption{Server-Consistency}
    \label{fig:pca_servercons}
  \end{subfigure}\hfill
  \begin{subfigure}[t]{0.32\textwidth}
    \centering
    \includegraphics[width=\linewidth]{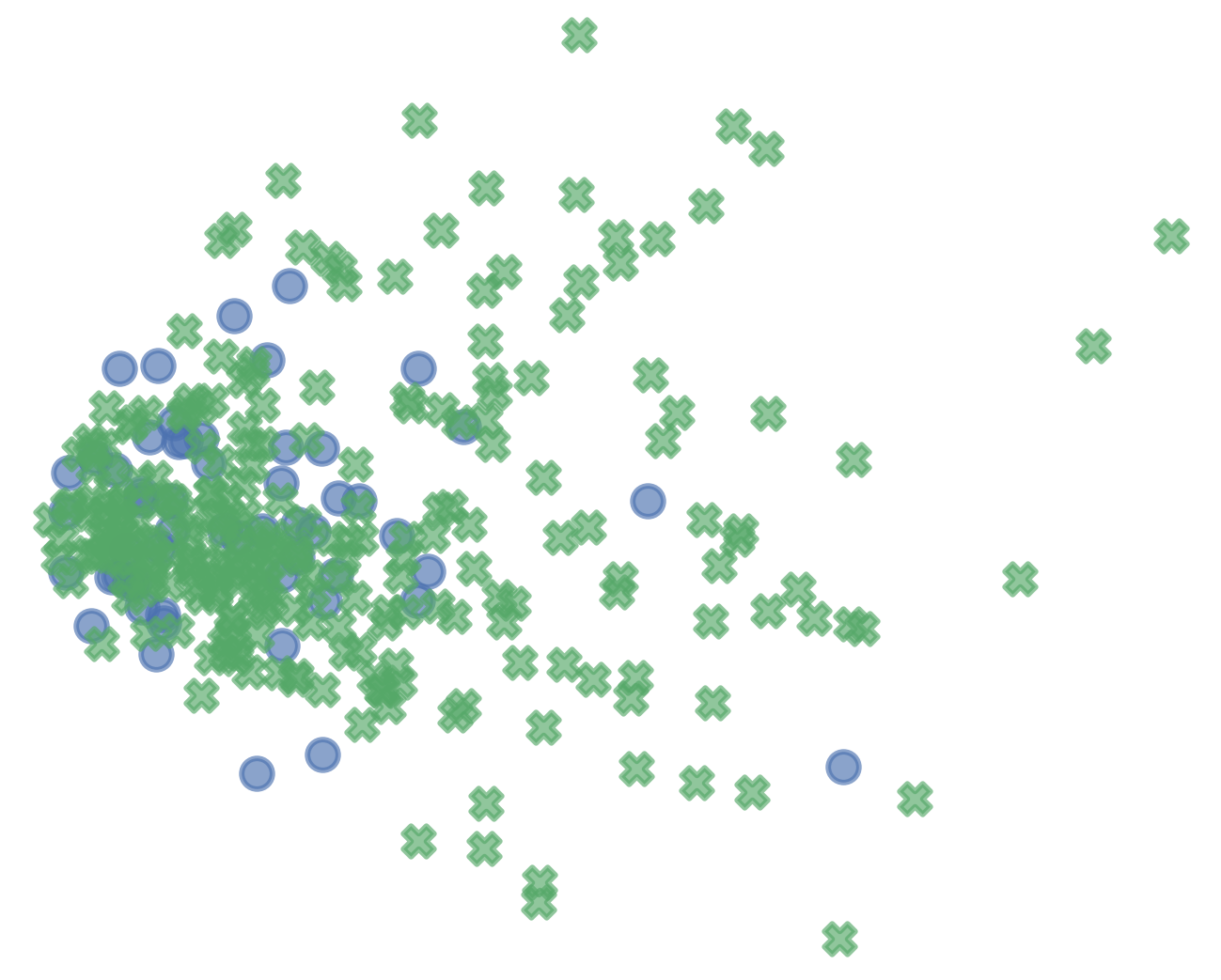}
    \caption{SCC--VFL (ours)}
    \label{fig:pca_sccvfl}
  \end{subfigure}

  \caption{PCA of fused representations colored by the sensitive-group indicator $s$ for all methods on German Credit (IID, seed 0), where $s=0$ and $s=1$ denote the two groups defined by the sensitive attribute. SCC--VFL shows the strongest mixing between the two groups while preserving a coherent manifold, suggesting reduced sensitive-attribute signal relative to the baselines.}
  \label{fig:pca_methods}
\end{figure*}

\subsection{Training dynamics of SCC-VFL}

Across German Credit IID experiments aggregated over 30 seeds, we inspect SCC-VFL training curves to verify that the gains reflect stable optimization rather than a single favorable run. Figure~\ref{fig:scc_training_curves} shows smooth loss convergence and consistently competitive validation accuracy, while SCG increases and then saturates once the consistency term activates and FR remains low, indicating SCC-VFL primarily refines logits and confidence under counterfactual edits instead of causing frequent label flips.

\begin{figure*}[t]
  \centering

  \begin{subfigure}[t]{0.32\textwidth}
    \centering
    \includegraphics[width=\linewidth, height=0.6\linewidth]{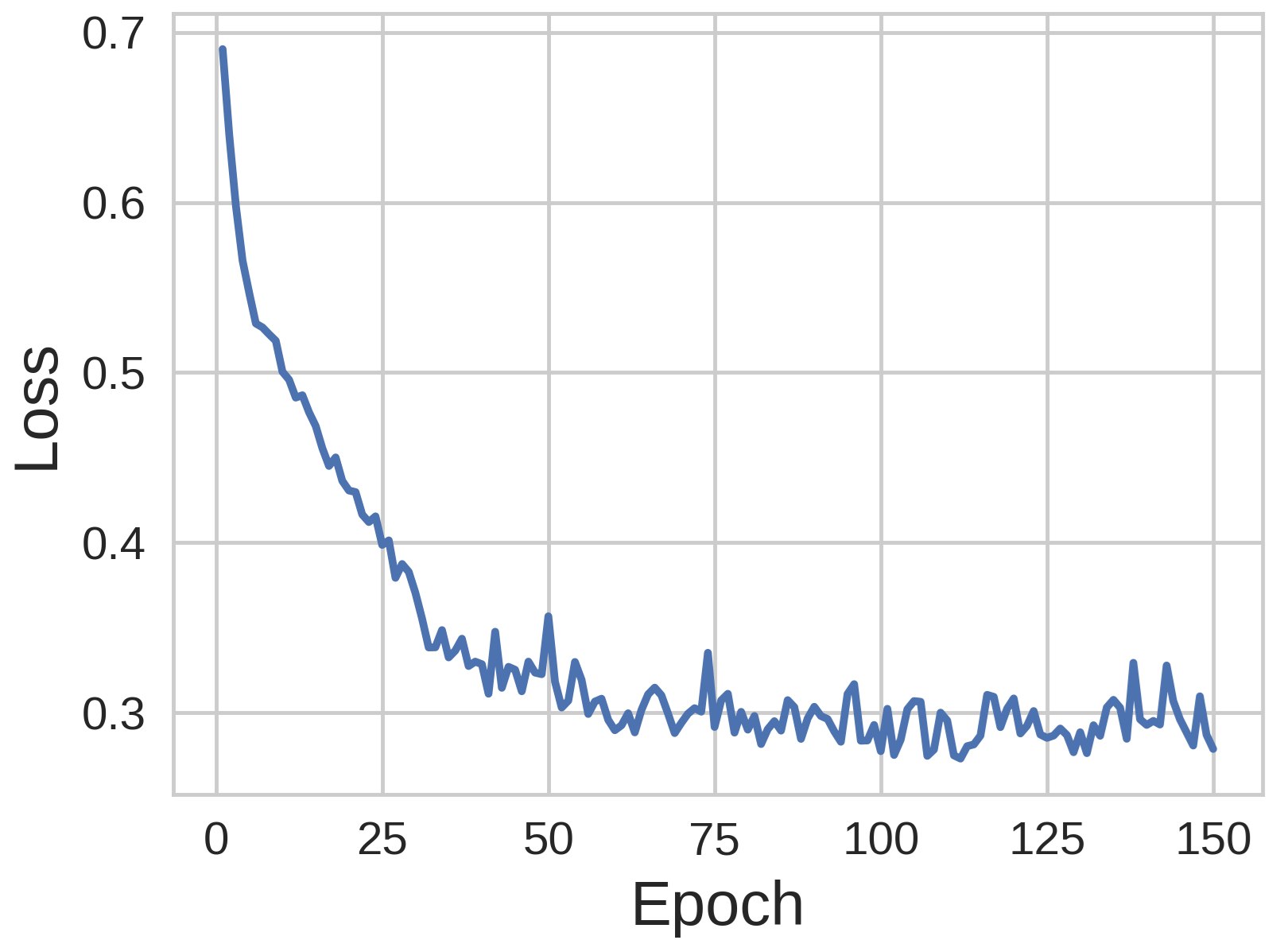}
    \caption{Training loss}
    \label{fig:sccvfl_trainloss}
  \end{subfigure}\hfill
  \begin{subfigure}[t]{0.32\textwidth}
    \centering
    \includegraphics[width=\linewidth, height=0.6\linewidth]{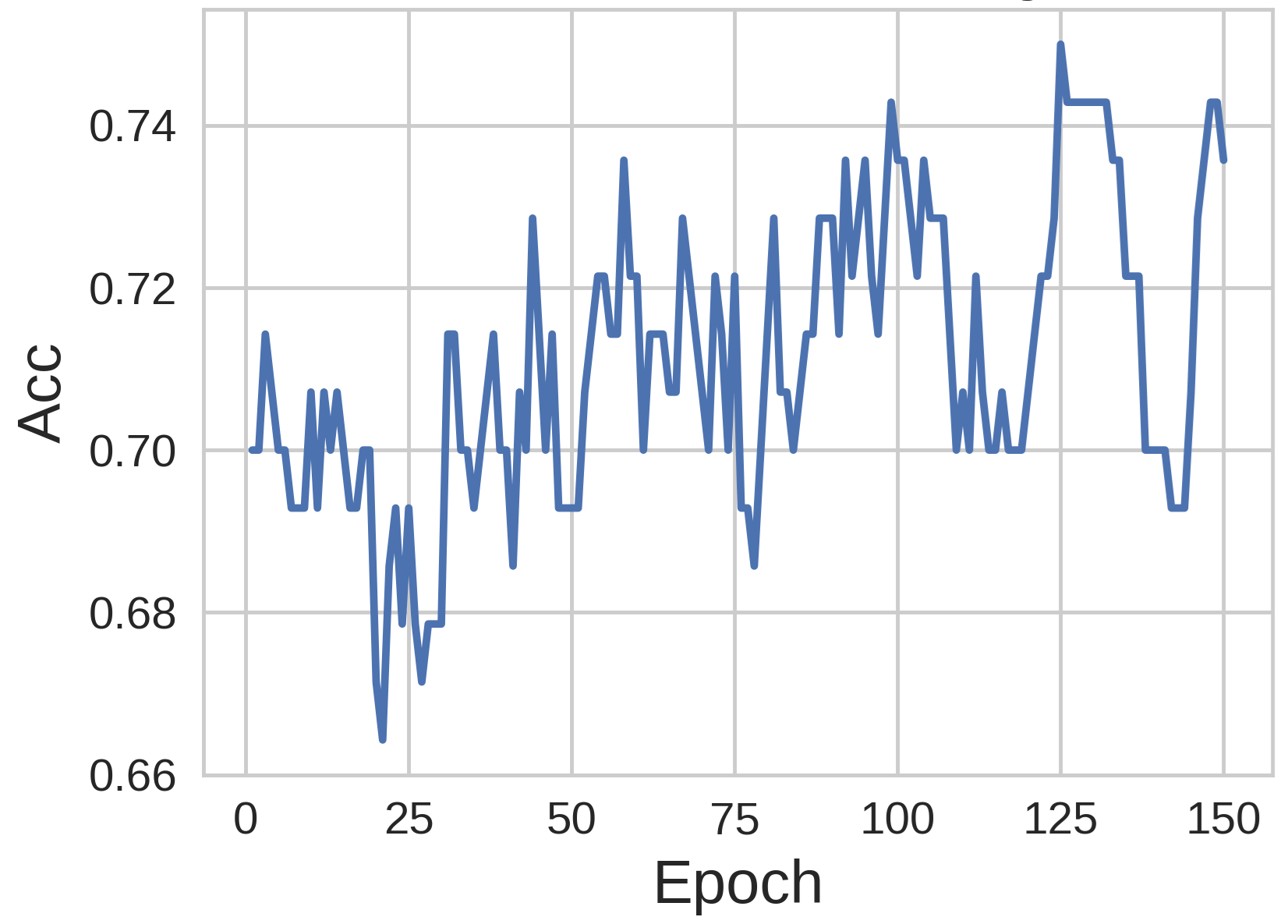}
    \caption{Validation accuracy}
    \label{fig:sccvfl_valacc}
  \end{subfigure}\hfill
  \begin{subfigure}[t]{0.32\textwidth}
    \centering
    \includegraphics[width=\linewidth, height=0.6\linewidth]{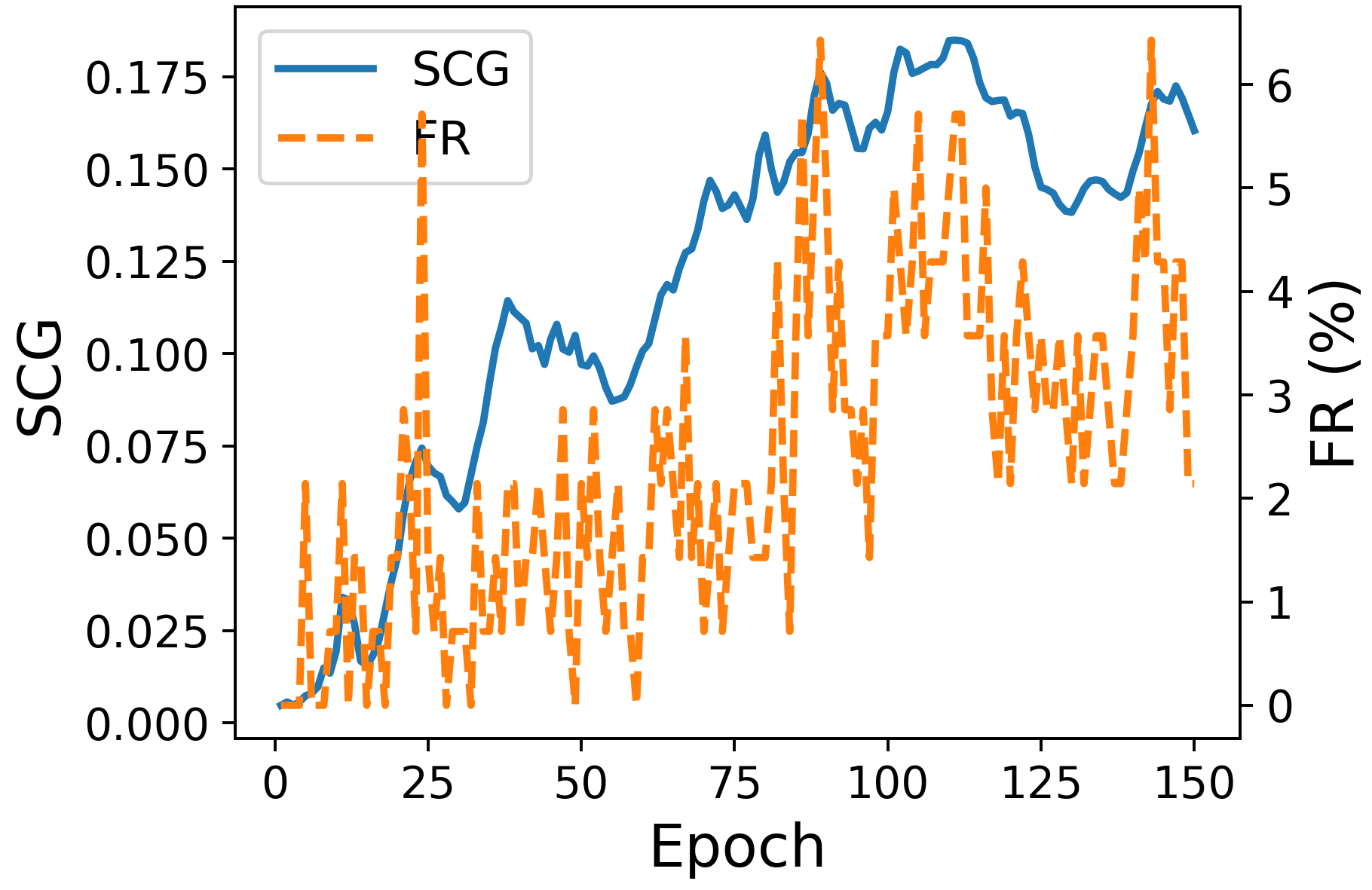}
    \caption{SCG and FR over epochs}
    \label{fig:scc-scg-fr}
  \end{subfigure}

  \caption{SCC--VFL training dynamics on German Credit IID: train loss, validation accuracy, and the evolution of SCG and FR over epochs. The model converges to a regime with low loss, strong accuracy, and low flip rate while maintaining non-trivial counterfactual consistency.}
  \label{fig:scc_training_curves}
\end{figure*}

\subsection{Counterfactual case studies across methods}

Instance-level summaries (seed 0) illustrating that SCC-VFL typically preserves labels with smaller, structured mediator edits, while some baselines require larger shifts and can flip labels on borderline cases.

\begin{tcolorbox}[title={Case study 1: correctly accepted applicant, \(y=0, s=1\)},fonttitle=\bfseries]
\small
Top standardized features: \texttt{credit\_amount}, \texttt{duration}, \texttt{purpose} \((M)\), \texttt{employment} \((P)\), \texttt{other\_install\_plans} \((M)\).  
Baseline and all fairness baselines keep the label \(y=0\) but Uniform-CF requires a large mediator shift \(\|\Delta_M\| \approx 0.52\). SCC-VFL keeps the label stable with a smaller and more structured change \(\|\Delta_M\| \approx 0.25\), mainly nudging employment and installment related mediators.
\end{tcolorbox}

\begin{tcolorbox}[title={Case study 2: ambiguous low-risk applicant, \(y=0, s=1\)},fonttitle=\bfseries]
\small
The baseline predicts \(y=0\) with low confidence, while Server-Consistency flips to \(y=1\) with high confidence, revealing residual dependence on \(s\). Uniform-CF and Policy-blind Mask leave \(y=0\) but change probabilities with large \(\|\Delta_M\|\). SCC-VFL preserves the decision and adjusts confidence in a moderate way, indicating that its generator learns a minimal recourse-style edit rather than over-correcting. This pattern is typical across similar borderline cases sampled over multiple seeds.
\end{tcolorbox}

\begin{tcolorbox}[title={Case study 3: high-risk applicant, \(y=1, s=1\)},fonttitle=\bfseries]
\small
Uniform-CF flips the label from \(1\) to \(0\) under strong mediator perturbations, which corresponds to a brittle counterfactual. Policy-blind Mask and Server-Consistency keep \(y=1\) but still rely on large or opaque changes. SCC-VFL maintains the positive decision with moderate mediator edits, confirming that it does not use counterfactuals to hide risk but instead enforces stability around the original decision.
\end{tcolorbox}

\subsection{Recourse-style mediator edits in SCC-VFL}

Seed-0 examples highlighting that SCC-VFL edits concentrate on mediator coordinates while leaving non-descendants fixed, producing small, policy-aligned directions of change reminiscent of actionable recourse~\cite{ustun2019actionable} and counterfactual explanations~\cite{wachter2017counterfactual}.

\begin{tcolorbox}[title={Recourse example A: low-risk, \( \hat{y}=0 \)},fonttitle=\bfseries]
\small
SCC-VFL slightly increases \texttt{employment} tenure and \texttt{housing} quality, and reduces \texttt{existing\_credits}, while keeping non-descendants such as \texttt{credit\_amount} and \texttt{duration} fixed. These changes resemble realistic recourse actions that a bank might recommend to improve future credit decisions without directly editing demographic attributes, and we observe similar patterns for other low-risk clients.
\end{tcolorbox}

\begin{tcolorbox}[title={Recourse example B: high-risk, \( \hat{y}=1 \)},fonttitle=\bfseries]
\small
For a predicted defaulter, SCC-VFL proposes small shifts in \texttt{housing}, \texttt{property}, and \texttt{other\_install\_plans}, again leaving non-descendants unchanged. The model preserves the current label but provides a concrete direction for improvement in mediator space, illustrating that SCC-VFL implements policy-aligned, on-support counterfactuals rather than arbitrary feature perturbations. Across random draws, mediator edits stay within similar magnitude ranges, which matches the low flip rates reported in the main tables.
\end{tcolorbox}

\section{Worked Example: Counterfactual Enforcement in Credit Decisions}
\label{app:worked_example}

This appendix provides an illustrative example of how SCC-VFL enforces selective
counterfactual consistency in a credit decision setting. The example is intended
to clarify the mechanics of mediator editing and server-side enforcement and does
not constitute a normative judgment about which features should be considered
permissible or impermissible in practice.

\noindent \textbf{Setup.}
Consider a loan applicant whose features are vertically partitioned across
institutions. The protected attribute is age, which is not provided as a model
input and is held by a trusted party. A policy review specifies the following
feature roles:
\begin{itemize}[nosep]
    \item Non-descendants $N$: loan amount and loan duration
    \item Permissible mediators $M$: employment tenure and credit history
    \item Impermissible proxies $P$: housing status and number of dependents
\end{itemize}
The applicant is initially observed with age corresponding to a younger group.

\noindent \textbf{Counterfactual generation.}
To evaluate the counterfactual scenario in which the applicant is older, the
server provides a target sensitive embedding corresponding to the older age
group. Each party applies the masked generator as follows. Non-descendants in
$N$ are copied exactly, ensuring identity on features that should not change
under the intervention. Mediators in $M$ are edited to remain on-support under
the conditional distribution for an older applicant, for example by increasing
employment tenure or modestly improving credit history. Proxy features in $P$
pass through unchanged but are guarded by the proxy adversary to suppress residual
sensitive signal.

\noindent \textbf{Server-side enforcement.}
The server fuses representations from the original and counterfactual inputs and
applies the Selective Counterfactual Consistency loss to penalize prediction
differences. Because only policy-permitted mediators are edited and proxy leakage
is controlled, this penalty targets impermissible influence of age on the lending
decision rather than suppressing legitimate pathways.

\noindent \textbf{Interpretation.}
Framed as a counterfactual explanation~\cite{wachter2017counterfactual}, if the predicted decision changes substantially between the original and counterfactual representations, the SCC penalty increases, discouraging age-driven instability. If the decision remains stable, the model satisfies selective counterfactual consistency for this individual under the given policy specification. This example illustrates how SCC-VFL operationalizes individual-level stability without requiring access to raw sensitive attributes or a fully specified causal graph.

\end{document}